\documentclass{LMCS}
\def\doi{9(1:09)2013}
\lmcsheading%
{\doi}
{1--27}
{}
{}
{Jun.~16, 2011}
{Mar.~\phantom04, 2013}
{}
\usepackage{hyperref,enumerate}
\usepackage{amsfonts}
\usepackage{amssymb}
\usepackage{amsmath}
\usepackage{latexsym}
\usepackage{graphicx}
\usepackage[curve,cmtip]{xypic}
\usepackage{verbatim}
\usepackage{multirow}
\mathcode`:="003A  % : ordinary symbol
\mathcode`;="003B  % ; ordinary symbol
\mathcode`?="003F  % ? ordinary symbol
\mathcode`|="026A  % | ordinary symbol
\mathcode`<="4268  % < abbreviates \langle
\mathcode`>="5269  % > abbreviates \rangle

\mathchardef\ls="213C    % less symbol (< used as \langle)
\mathchardef\gr="213E    % greater symbol (> used as \rangle)
\mathchardef\uparrow="0222  % adaptation mathmode of uparrow
\mathchardef\downarrow="0223  % adaptation mathmode of downarrow

\def\tr#1{\stackrel{#1}{\to}}
\def\pow#1{{\mathcal P_\omega}#1}

\def\T{\mathbf{T}}

\def\Set{\mathbf{Set}}

\newcommand{\bb}[1]{[\![ #1 ]\!]}
\newtheorem*{remark}{Remark}
\begin{document}
%
% For FIXME comments
%
%\FXRegisterAuthor{mb}{amb}{Mar}%Marcello
%\FXRegisterAuthor{fb}{asm}{Fil}%Filippo command: \fbnote{bla bla}
%\FXRegisterAuthor{as}{abd}{Alex}%Alexandra
%\FXRegisterAuthor{jr}{aas}{Jan}%Jan

%\title{Determinizations for a large class of automata}
\title{Generalizing determinization from automata to coalgebras}

\author[A.~Silva]{Alexandra Silva\rsuper a}   %required
\address{{\lsuper a}Radboud University Nijmegen and Centrum Wiskunde \& Informatica} %required
\email{ams@cwi.nl}  %optional
\thanks{{\lsuper a}The work of Alexandra Silva is partially funded by the ERDF through the Programme COMPETE and by the Portuguese Foundation for Science and
Technology, project ref. \texttt{PTDC/EIA-CCO/122240/2010} and \texttt{SFRH/BPD/71956/2010}.
}

\author[F.~Bonchi]{Filippo Bonchi\rsuper b}   %required
\address{{\lsuper b}ENS Lyon, Universit\'{e} de Lyon, LIP (UMR 5668 CNRS ENS Lyon UCBL INRIA)} %required
\email{filippo.bonchi@ens-lyon.fr}  %optional
\thanks{{\lsuper b}The work of Filippo Bonchi is supported by the CNRS PEPS project CoGIP and the project ANR 12IS02001 PACE}   %optional

\author[M.~Bonsangue]{Marcello Bonsangue\rsuper c}   %required
\address{{\lsuper c}LIACS - Leiden University} %required
\email{marcello@liacs.nl}  %optional
%\thanks{Supported by the CNRS PEPS project CoGIP}   %optional

\author[J.~Rutten]{Jan Rutten\rsuper d}   %required
\address{{\lsuper d}Centrum Wiskunde \& Informatica and Radboud University Nijmegen} %required
\email{janr@cwi.nl}  %optional
\thanks{{\lsuper{c,d}}The research of Marcello Bonsangue and Jan Rutten has been carried out under the Dutch NWO project {\em CoRE: Coinductive Calculi for Regular Expressions.},
dossier number 612.063.920.}   %optional
%\thanks{Supported by the CNRS PEPS project CoGIP}   %optional

\keywords{Coalgebras, Powerset Construction, Linear Semantics} 
\ACMCCS{[{\bf Theory of computation}]: Models of
  computation---Abstract machines \& Formal languages and automata
  theory---Formalisms---Algebraic language theory \& Semantics and reasoning---Program semantics---Categorical semantics}
\subjclass{F.3.2}
%\titlecomment{This paper is an extended version of \cite{FSTTCS2010}.}

\begin{abstract}
The powerset construction is a standard method for converting a
nondeterministic automaton into a deterministic one recognizing the 
same language. In this paper, we lift the powerset construction
from automata to the more general framework of coalgebras with structured
state spaces. Coalgebra is an abstract framework for the uniform study
of different kinds of dynamical systems. An endofunctor $F$ determines
both the type of systems ($F$-coalgebras) and a notion of behavioural
equivalence ($\sim_F$) amongst them. Many types of transition systems and
their equivalences can be captured by a functor $F$. For example, for
deterministic automata the derived equivalence is language equivalence,
while for non-deterministic automata it is ordinary bisimilarity.

We give several examples of applications of our generalized determinization 
construction, including partial Mealy machines, (structured) Moore automata, Rabin
probabilistic automata, and, somewhat surprisingly, even pushdown automata.
%A somewhat  surprising outcome of our framework is the first coalgebraic characterization
%of pushdown automata, a long standing open question in the coalgebraic community.
To further witness the generality of the approach we show how to characterize
coalgebraically several equivalences which have been object of interest in
the concurrency community, such as failure or ready semantics.
\end{abstract}
\maketitle

\section*{Introduction}

\emph{Coalgebra} is by now a well established general framework for
the study of the behaviour of large classes of dynamical systems,
including various kinds of automata (deterministic, probabilistic
etc.) and infinite data types (streams, trees and the like). For a
functor  $F\colon \Set \rightarrow \Set$, an $F$-coalgebra is a pair
$(X, f)$, consisting of a set $X$ of states and a function $f \colon
X \rightarrow F(X)$ defining the observations and transitions of the
states. Coalgebras generally come equipped with a standard notion of
equivalence called \emph{$F$-behavioural equivalence} that is fully
determined by their (functor) type $F$. Moreover, for most functors
$F$ there exists a \emph{final} coalgebra into which any
$F$-coalgebra is mapped by a unique homomorphism that identifies all
$F$-equivalent states.

Much of the coalgebraic approach can be nicely illustrated with
deterministic automata (DA), which are coalgebras of the functor
\mbox{$D(X) = 2 \times X^A$}. In a DA, two states are $D$-equivalent
precisely when they accept the same language. The set $2^{A^*}$ of
all formal languages constitutes a final $D$-coalgebra, into which
every DA is mapped by a homomorphism that sends any state to the
language it accepts.

It is well-known that \emph{non-deterministic} automata (NDA)
often provide more efficient (smaller) representations of
formal languages than DA's. Language acceptance of NDA's is typically
defined by turning them into DA's via the \emph{powerset construction}.
Coalgebraically this works as follows.
NDA's are coalgebras of the functor $N(X) = 2 \times \pow(X)^A$, where
$\pow$ is the finite powerset. An $N$-coalgebra $(X, f \colon X \to 2 \times \pow(X)^A)$
is \emph{determinized} by transforming it into a $D$-coalgebra
$(\pow(X), f^\sharp \colon \pow(X) \to 2 \times \pow(X)^A)$
(for details see Section~\ref{sec:motiv}).
Then, the language accepted by a state $s$ in the NDA $(X,f)$
is defined as the language accepted by the state $\{s\}$ in the DA
$(\pow(X), f^\sharp )$.

For a second variation on DA's, we look at
\emph{partial automata} (PA): coalgebras of the functor
$P(X) = 2 \times (1+X)^A$, where for certain input letters
transitions may be undefined. Again, one is often interested
in the DA-behaviour (i.e., language acceptance) of PA's. This
can be obtained  by turning them into DA's using \emph{totalization}.
Coalgebraically, this amounts to the transformation of a $P$-coalgebra
$(X, f \colon X \to 2 \times (1+X)^A )$ into a $D$-coalgebra
$(1+X, f^\sharp \colon 1+X \to 2 \times (1+X)^A )$.

Although the two examples above may seem very different, they
are both instances of one and the same phenomenon, which it is
the goal of the present paper to describe at a general level.
Both with NDA's and PA's, two things happen at the
same time: (i) more (or, more generally, different types of)
transitions are allowed, as a consequence of changing
the functor type by replacing $X$ by $\pow(X)$ and $(1+ X)$, respectively;
and (ii) the behaviour of NDA's and PA's is still given in
terms of the behaviour of the original DA's (language acceptance).

For a large family of $F$-coalgebras, both (i) and (ii) can be
captured simultaneously with the help of the categorical
notion of \emph{monad}, which generalizes the notion of algebraic theory.
The structuring of the state space $X$ can be expressed as a change of
functor type from $F(X)$ to $F(T(X))$. In our examples above, both the functors
$T_1(X) = \pow(X)$ and $T_2(X) = 1+X$ are monads,
and NDA's and PA's are obtained from DA's by changing the
original functor type $D(X)$ into $N(X) = D(T_1(X))$ and $P(X) = D(T_2(X))$.
Regarding (ii), one assigns $F$-semantics to an $FT$-coalgebra
$(X,f)$  by transforming it into an $F$-coalgebra $(T(X),f^\sharp)$,
again using the monad $T$. In our examples above,
the determinization of NDA's and the totalization of PA's
consists of the transformation of  $N$- and $P$-coalgebras
$(X,f)$ into $D$-coalgebras $(T_1(X), f^\sharp)$ and $(T_2(X), f^\sharp)$,
respectively.

We shall investigate general conditions on the functor types under
which the above constructions can be applied: for one thing, one has
to ensure that the  $FT$-coalgebra map $f\colon X \to F(T(X))$ induces a suitable
$F$-coalgebra map $f^\sharp \colon T(X) \to F(T(X)) $. 
Our results will lead to a uniform
treatment of all kinds of existing and new variations of automata,
that is, $FT$-coalgebras, by an algebraic structuring of their state
space through a monad $T$. Furthermore, we shall prove a number of
general properties that hold in all situations similar to the ones
above. For instance, there is the notion of $N$-behavioural
equivalence with which NDA's, being $N$-coalgebras, come equipped.
It coincides with the well-known notion of Park-Milner bisimilarity
from process algebra. A general observation is that if two states in
an NDA are $N$-equivalent then they are also $D$- (that is,
language-) equivalent. For PA's, a similar statement holds. One
further contribution of this paper is a proof of these statements,
once and for all for all $FT$-coalgebras under consideration.

Coalgebras of type $FT$ were studied in~\cite{Lenisa99,bartels,jacobs05}. In~\cite{bartels,jacobs05}
the main concern was definitions by coinduction, whereas in~\cite{Lenisa99} a proof principle was also presented.
 All in all, the present paper can be seen as the understanding of the aforementioned papers from a new perspective, presenting a uniform view on various automata constructions and equivalences.

The structure of the paper is as follows. After preliminaries
(Section~\ref{sec:prelim}) and
the details of the motivating examples above (Section~\ref{sec:motiv}),
Section~\ref{sec:general} presents the general construction
 as well as many more examples, including the coalgebraic chracterisation of pushdown automata (Section~\ref{pda}).
In Section~\ref{sec:bisim_implies_trace}, a large family of automata
(technically: functors) is characterised to which the constructions
above can be applied. Section~\ref{secbbat} contains the application of the framework in order to recover several 
interesting equivalences stemming from the world of concurrency, such as failure and ready semantics. 
Section~\ref{sec:discussion} discusses related work and presents pointers to future work.

This paper is an extended version of~\cite{FSTTCS}. Compared to the conference version, we include the proofs and more examples. More interestingly, the characterisation of pushdown automata coalgebraically (Section~\ref{pda}) and the
material in Section~\ref{secbbat} are original.

%In Appendix \ref{secbbat}, we further prove the expressivity of our
%framework by showing that it can subsume many behavioural
%equivalences from the so called \emph{linear-time branching-time
%spectrum}~\cite{Glabbeek90}.

%
\section{Background}~\label{sec:prelim}
In this section we introduce the preliminaries on coalgebras and
algebras. First, we fix some notation on sets. We will denote sets
by capital letters $X,Y,\ldots$ and functions by lower case letters
$f,g,\dots$ Given sets $X$ and $Y$, $X\times Y$ is the cartesian
product of $X$ and $Y$ (with the usual projection maps $\pi_1$ and
$\pi_2$), $X+Y$ is the disjoint union (with injection maps
$\kappa_1$ and $\kappa_2$) and $X^Y$ is the set of functions $f\colon Y\to
X$. The collection of finite subsets of $X$ is denoted by $\pow
(X)$, while the collection of full-probability distributions with finite
support is $\mathcal D_\omega (X)
= \{f\colon X\to [0,1] \mid \text{$f$ finite support and } \sum_{x\in X} f(x) = 1\}$. For a set of letters
$A$, $A^*$ denotes the set of all words over $A$; $\epsilon$ the
empty word; and $w_1\cdot w_2$ (and $w_1w_2$) the concatenation of
words $w_1,w_2 \in A^*$.

\subsection{Coalgebras} A coalgebra is a pair $(X, f\colon
X\to F(X))$, where $X$ is a set of states and $F\colon
\mathbf{Set}\to \mathbf{Set}$ is a functor.
 The functor $F$, together with the
function $f$, determines the {\em transition structure} (or
dynamics) of the $F$-coalgebra~\cite{Rutten00}.

An {\em $F$-homomorphism\/} from an $F$-coalgebra $(X,f)$ to an
$F$-coalgebra $(Y,g)$ is a function $h\colon \, X \to Y$ preserving the
transition structure, {\em i.e.}, $g\circ h = F(h) \circ f$.

An $F$-coalgebra $(\Omega,\omega)$ is said to be {\em final} if for
any $F$-coalgebra $(X,f)$ there exists a unique $F$-homomorphism
 $\bb{-}_X\colon X\to \Omega$.
 All the functors considered in examples in this paper have a final coalgebra.

Let $(X, f)$ and $(Y,g)$ be two $F$-coalgebras. We say that the
states $x\in X$ and $y\in Y$ are {\em behaviourally equivalent},
written $x\sim_F y$, if and only if they are mapped into the same
element in the final coalgebra, that is $\bb{x}_X =
\bb{y}_Y$.%\footnote{Bisimilarity is usually defined in the literature
%in a slight different way. The definition we present here is often
%called {\em behavioural equivalence}. For most functors both notions
%coincide and we choose the notion of equivalence which is more
%convenient for presenting our story.}.

For weak pullback preserving functors, behavioural equivalence coincides with the usual notion 
of bisimilarity~\cite{Rutten00}.
   
\subsection{Algebras}
Monads can be thought of as a generalization of algebraic theories.
A \emph{monad} $\T = (T,\mu,\eta)$ is a triple consisting of an
endofunctor $T$ on $\Set$ and two natural transformations: a
\emph{unit} $\eta\colon \mathit{Id} \Rightarrow T$ %mapping a set $X$
%to its free algebra $T(X)$,  
and a \emph{multiplication} $\mu\colon
T^2 \Rightarrow T$.
They satisfy the following commutative laws
%\vspace*{-.3cm}
\[
\mu \circ \eta_T = id_T = \mu \circ T\eta
\;\;\;\mbox{and}\;\;\;
\mu \circ \mu_T = \mu \circ T\mu.
\]
Sometimes it is more convenient to represent a monad $\T$, equivalently, as a \emph{Kleisli triple}
$(T, (\_)^\sharp, \eta)$~\cite{Man76}, where $T$ assigns a set $T(X)$ to each
set $X$, the unit $\eta$ assigns a function $\eta_X \colon X \to T(X)$  to each set $X$,
and the extension operation $(\_)^\sharp$  assigns  to each
$f \colon X \rightarrow T(Y)$ a function $f^\sharp \colon T(X) \rightarrow T(Y)$, such that,
%\vspace*{-.3cm}
\[
f^\sharp \circ \eta_X = f
\;\;\;\;\;\;
(\eta_X)^\sharp = id_{T(X)}
\;\;\;\;\;\;
(g^\sharp \circ f)^\sharp
= g^\sharp \circ f^\sharp \,,
\]
for $g\colon Y  \rightarrow T(Z)$. Monads are frequently referred to
as \emph{computational types}~\cite{Moggi}. We list now a few examples. In what follows, $f\colon X\to T(Y)$ and $c\in T(X)$.

\paragraph{\textbf{Nondeterminism}} $T(X) = \pow(X)$; $\eta_X$ is the singleton map $x\mapsto \{x\}$;  $f^\sharp (c) = \bigcup_{x\in c} f(x)$.
\paragraph{\textbf{Partiality}} $T(X)  = 1+ X $ where $1=\{*\}$ represents a terminating (or diverging) computation; $\eta_X$ is the injection map $\kappa_2 \colon X\to 1+X$; $f^\sharp(\kappa_1(*)) = \kappa_1(*)$ and $f^\sharp(\kappa_2(x)) = f(x)$.

\smallskip

\noindent\textbf{Further examples} of monads include: exceptions ($T(X)= E +X$), side-effects
($T(X) = (S \times X)^S$), interactive output ($T(X) =  \mu v. X +
(O\times v) \cong O^*\times X$) and full-probability ($T(X) = \mathcal
D_\omega (X)$). We will use all these monads in our examples and we
will define $\eta_X$ and $f^\sharp$ for each later in
Section~\ref{sec:examples}.

\medskip

A $\T$\emph{-algebra} of a monad $\T$ is a pair $(X,h)$ consisting of a set $X$, called
carrier, and a function $h\colon T(X) \rightarrow X$  such that $h \circ \mu_X = h \circ Th$
and $h \circ \eta_X = id_X$. A $T$-homomorphism between two $\T$-algebras $(X,h)$ and $(Y,k)$
is a function $f\colon X \to Y$ such that $f \circ h = k \circ Tf$. $\T$-algebras and their
homomorphisms form the so-called \emph{Eilenberg-Moore category}
$\Set^\T$. There is a forgetful
functor $U^\T\colon\Set^\T \to \Set$ defined by
%\vspace*{-.2cm}
\[
U^\T((X,h)) = X
\;\;\;\mbox{and}\;\;\;
U^\T(f\colon(X,h)\rightarrow (Y,k)) = f\colon X \rightarrow Y \,.
\]
%\vspace*{-.6cm}

The forgetful functor $U^\T$ has left adjoint $X \mapsto (T(X),\mu_X\colon TT(X) \to T(X))$,
mapping a set $X$ to its free $\T$-algebra. If $f \colon X \rightarrow Y$ with $(Y,h)$
a $\T$-algebra, the unique $\T$-homomorphism $f^\sharp \colon (T(X), \mu_X) \rightarrow (Y, h)$
with $f^\sharp \circ \eta_X = f$ is given by
%\vspace*{-.3cm}
\[
\xymatrix{
f^\sharp\colon T(X) \ar[r]^{Tf} & T(Y) \ar[r]^{h} & Y \,.
}
\]
%\vspace*{-.6cm}

The function $f^\sharp\colon (T(X), \mu_X) \rightarrow (T(Y), \mu_Y)$ coincides with
function extension for a Kleisli triple.
%The forgetful functor $U^\T$ creates and preserves all small limits and filtered colimits.
%
For the monad $\pow$ the associated Eilenberg-Moore category is the category of
join semi-lattices, whereas for the monad $1+-$ is the category of pointed sets.
\section{Motivating examples}~\label{sec:motiv}
In this section, we introduce two motivating examples. We will present two constructions,  the determinization of a non-deterministic
automaton and the totalization of a partial automaton, which we will later show to be an instance of the same, more general, construction.

\subsection{Non-deterministic automata}
A deterministic automaton (DA) over the input alphabet $A$ is a pair
$(X,<o,t>)$, where $X$ is a set of states and $<o,t> \colon X \to
2\times X^A$ is a function with two components: $o$, the output
function, determines if a state $x$ is final ($o(x) = 1$) or not
($o(x) = 0$); and $t$, the transition function, returns for each
input letter $a$ the next state. DA's are coalgebras for the functor
$2\times \mathit{Id}^A$. The final coalgebra of this functor is
$(2^{A^*},<\epsilon, (-)_a>)$ where $2^{A^*}$ is the set of
languages over $A$ and $<\epsilon, (-)_a>$, given a language $L$,
determines whether or not the empty word is in the language
($\epsilon(L) =1$ or $\epsilon(L)=0$, resp.) and, for each input
letter $a$, returns the {\em derivative} of $L$: $L_a = \{ w \in A^*
\mid aw\in L\}$.
From any DA, there is a unique map $l$ into $2^{A^*}$ which assigns
to each state its behaviour (that is, the language that the state
recognizes).
%\vspace*{-.4cm}
\[
\xymatrix@C=2cm@R=1.3cm{X \ar@{-->}[r]^{l}\ar[d]_{<o,t>} &
2^{A^*}\ar[d]^{<\epsilon, (-)_a>}\\
2\times X^A\ar@{-->}[r]_-{\mathit{id}\times l^A} & 2\times (2^{A^*})^A }
\]
A non-deterministic automaton (NDA) is similar to a DA but the
transition function gives a set of next-states for each input letter
instead of a single state. Thus, an NDA over the input alphabet $A$
is a pair $(X,<o,\delta>)$, where $X$ is a set of states and
$<o,\delta> \colon X \to 2\times (\pow(X))^A$ is a pair of functions
with $o$ as before and where $\delta$ determines for each input
letter $a$ a set of possible next states. In order to compute the
language recognized by a state $x$ of an NDA $\mathcal A$, it is
usual to first determinize it, constructing a DA
$\mathbf{det}(\mathcal A)$ where the state space is $\pow(X)$, and
then compute the language recognized by the state $\{x\}$ of
$\mathbf{det}(\mathcal A)$. Next, we describe in coalgebraic terms
how to construct the automaton $\mathbf{det}(\mathcal A)$.

%
%
%
%FILIPPO: Here I have specified the type of the functions
Given an NDA $\mathcal A = (X, <o,\delta>)$, we construct
$\mathbf{det}(\mathcal A) = (\pow(X), <\overline o,t>)$, where, for
all $Y\in \pow(X)$, $a\in A$, the functions $\overline o\colon\pow(X) \to
2$ and $t\colon\pow(X) \to \pow(X)^A$ are
$$\overline o (Y) = \begin{cases} 1 & \exists_{y\in Y}
o(y) =1\\ 0 &\text{otherwise}\end{cases} \qquad t(Y)(a) =
\bigcup\limits_{y\in Y}\delta(y)(a)\text{.}$$ 
(Observe that these definitions exploit the join-semilattice structures of $2$ and $\pow(X)^A$).

The automaton
$\mathbf{det}(\mathcal A)$  is such that the language $l(\{x\})$
recognized by $\{x\}$ is the same as the one recognized by $x$ in
the original NDA $\mathcal A$ (more generally, the language
recognized by state $X$ of $\mathbf{det}(\mathcal A)$ is the union
of the languages recognized by each state $x$ of $\mathcal A$).

%OLD VERSION
%Given an NDA $\mathcal A = (S, <o,\delta>)$, we construct
%$\mathbf{det}(\mathcal A) = (\pow(S), <\overline o,t>)$, where, for
%$X\subseteq S$,
%$$\overline o (X) = \begin{cases} 1 & \exists_{x\in X}
%o(x) =1\\ 0 &\text{oth.}\end{cases}$$ and, for each input letter
%$a$, $t(X)(a) = \bigcup\limits_{x\in X}\delta(x)(a)$. The automaton
%$\mathbf{det}(\mathcal A)$  is such that the language $l(\{s\})$
%recognized by $\{s\}$ is the same as the one recognized by $s$ in
%the original NDA $\mathcal A$ (more generally, the language
%recognized by state $X$ of $\mathbf{det}(\mathcal A)$ is the union
%of the languages recognized by each state $x$ of $\mathcal A$).

We summarize the situation above with the following commuting diagram:
%\vspace*{-.4cm}
\[
\xymatrix@R=1.5cm{ X \ar[d]_-{<o,\delta>}\ar@{->}[r]^-{\{\cdot\}}&
\pow(X)\ar[dl]^-{<\overline o,t>}\ar@{-->}^{l}[r] & 2^{A^*}\ar[d]^{<\epsilon, (-)_a>}\\
2\times \pow(X)^A\ar@{-->}[rr]_-{\mathit{id} \times l^A} && 2\times (2^{A^*})^A }
\]
We note that the language semantics of NDA's,
presented in the above diagram, can also be obtained as an instance of
the abstract definition scheme of $\lambda$-coinduction~\cite{bartels,jacobs05}.
\subsection{Partial automata}
A partial automaton (PA) over the input alphabet $A$ is a pair
$(X,<o,\partial>)$ consisting of a set of states $X$ and a pair of
functions $<o,\partial> \colon X \to 2\times (1+X)^A$.
Here $o\colon X \to
2$ is the same as with DA. The second function $\partial \colon X\to (1+X)^A$ is
a transition function that sends any state $x \in X$ to a function
$\partial(x)\colon A \to 1+X$, which for any input letter $a \in A$
is either undefined (no $a$-labelled transition takes place)
or specifies the next state that is reached.
 PA's are coalgebras for the functor
$2\times (1+\mathit{Id})^A$.  Given a PA $\mathcal A$, we can
construct a total (deterministic) automaton
$\mathbf{tot}(\mathcal A)$ by adding an extra {\em sink} state to
the state space: every undefined $a$-transition from a state $x$ is then
replaced by a $a$-labelled transition from $x$ to the sink state.
More precisely, given a  PA $\mathcal A =(X,<o,\partial>)$, we construct
$\mathbf{tot}(\mathcal A) = (1+X, <\overline o, t>)$,
where
\[
\begin{array}{l}
\overline o(\kappa_1(*)) = 0\\
 \overline o(\kappa_2(x)) = o(x)\\
\end{array}\hspace{1cm}
\begin{array}{l}
t(\kappa_1(*))(a) = \kappa_1(*)\\
t(\kappa_2(x))(a) = \partial(x)(a)\\
\end{array}
\]
(Observe that these definitions exploit the pointed-set structures of $2$ and $1+X$).

The language $l(x)$ recognized by a state $x$ will be precisely the
language recognized by $x$ in the original partial automaton.
Moreover, the new sink state recognizes the empty language. Again we
summarize the situation above with the help of following commuting
diagram, which illustrates the similarities between both
constructions:
%\vspace*{-.2cm}
\[
\xymatrix@R=1.5cm{ X \ar[d]_-{<o,\partial>}\ar@{->}[r]^-{\kappa_2}&
1+X\ar[dl]^-{<\overline o,{t}>}\ar@{-->}^{l}[r] & 2^{A^*}\ar[d]^{<\epsilon, (-)_a>}\\
2\times (1+X)^A\ar@{-->}[rr]_-{\mathit{id} \times l^A}  && 2\times (2^{A^*})^A }
\]
\section{Algebraically structured coalgebras}\label{sec:general}

In this section we present a general framework where both motivating
examples can be embedded and uniformly studied. We will consider
coalgebras for which the functor type~$FT$ can be decomposed into a
transition type $F$ specifying the relevant dynamics of a system and
a monad $T$ providing the state space with an algebraic structure.
For simplicity, we fix our base category to be $\Set$.%, but all
%results below can be generalized to an arbitrary category $\mathbf
%C$ with enough limits.

%Let $\T$ be a monad and $f\colon X \to FT(X)$ be a coalgebra, such
%that $FT(X)$ is a $\T$-algebra, that is, $FT(X)$ is the carrier of a
%$\T$-algebra $(FT(X),h)$.
%In the motivating examples this would amount to require that: for
%non-deterministic automata,  $2\times \pow(X)^A$ is a join-semi
%lattice; for partial automata, $2\times (1+X)^A$ is a pointed set.
%Both are indeed the case, as long as one looks at $2$ as a join semi
%lattice and as a pointed set (with distinguished element 0),
%respectively.

We study coalgebras $f\colon X \to FT(X)$ for a functor $F$ and a
monad $\T$ such that $FT(X)$ is a $\T$-algebra, that is $FT(X)$ is the carrier of a $\T$-algebra
$(FT(X),h)$. In the motivating examples, $F$ would be instantiated to
$2\times Id^A$ (in both) and $T$ to $\pow$, for NDAs, and to $1+-$ for
PAs. The condition that $FT(X)$ is a $\T$-algebra would amount to
require that $2\times \pow (X)^A$ is a join-semilattice, for NDAs, and
that $2\times (1+X)^A$ is a pointed set, for PAs. This is indeed the
case, since the set $2$ can be regarded both as a join-semilattice
($2\cong \pow(1)$) or as a pointed set ($2\cong 1+1$) and, moreover,
products and exponentials preserve the algebra structure.

The inter-play between the transition type $F$ and the
computational type $\T$ (more precisely, the fact that $FT(X)$ is a $\T$-algebra) allows each coalgebra $f\colon X
\rightarrow FT(X)$ to be extended uniquely to a $T$-algebra
morphism $ f^\sharp\colon(T(X),\mu_X) \rightarrow (FT(X),h)$ which
makes the following diagram commute.
%\vspace*{-.5cm}
\[
\xymatrix@R=1.5cm{
X \ar[d]_-{f}\ar[r]^-{\eta_X} & T(X)\ar[dl]^-{f^\sharp}\\
FT(X) &
}\ \ \ \ \ \ \  f^\sharp \circ \eta_X = f
\]
Intuitively, $\eta_X\colon X \rightarrow T(X)$ is the
inclusion of the state space of the coalgebra $f\colon X \rightarrow
FT(X)$ into the structured state space $T(X)$, and $f^\sharp\colon T(X)
\rightarrow FT(X)$ is the extension of the coalgebra $f$ to $T(X)$.

Next, we study the behaviour of a given state or, more
generally, we would like to say when two states $x_1$ and $x_2$ are
equivalent. The obvious choice for an equivalence would be
$FT$-behavioural equivalence. However, this equivalence is not
exactly what we are looking for. In the motivating example of
non-deterministic automata we wanted two states to be equivalent if
they recognize the same language. If we would take the equivalence
arising from the functor $2\times \pow(\mathit{Id})^A$ we would be
distinguishing states that recognize the same language but have
difference branching types, as in the following example.
%\vspace{-.5cm}
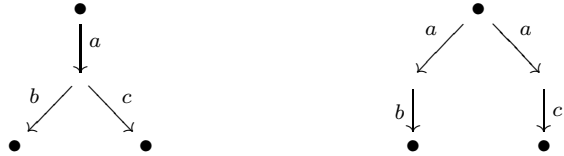
\[
\xymatrix@R=0.6cm@C=0.45cm{&\bullet \ar[d]^{a} &&\hspace{2cm} && \bullet \ar[dr]^{a}\ar[dl]_{a}\\
 &\ar[dr]^{c}\ar[dl]_{b}& && \ar[d]_{b}&&\ar[d]^{c}\\
 \bullet && \bullet &&\bullet&&\bullet }
\]
We now define a new equivalence, which {\em absorbs} the effect of the monad $T$.

We say that two elements $x_1$ and $x_2$ in $X$ are
\emph{$F$-equivalent with respect to a monad} $\T$, written $x_1
\approx_F^T x_2$, if and only if $\eta_X(x_1) \sim_F \eta_X(x_2)$.
The equivalence $\sim_F$ is just $F$-behavioural equivalence for the
$F$-coalgebra $f^\sharp \colon T(X) \to FT(X)$.

If the functor $F$ has a final coalgebra  $(\Omega, \omega)$ , we can capture the semantic equivalence above in the following commuting diagram
%\vspace{-.5cm}
\begin{eqnarray}
\label{F-final}
\xymatrix@C=2cm@R=1.5cm{
X \ar[d]_{f}\ar[r]^-{\eta_X} & T(X)\ar[dl]^-{f^\sharp}\ar@{-->}[r]^{\bb{-}} & \Omega \ar[d]^{\omega}\\
FT(X)\ar@{-->}[rr]_-{F\bb{-}}  && F(\Omega)
}
\end{eqnarray}
Returning to our first example, two states $x_1$ and $x_2$ of an NDA (in
which $T$ is instantiated to $\pow$ and $F$ to
$2\times\mathit{Id}^A$) would satisfy $x_1\approx_F^T x_2$ if and
only if they recognize the same language (recall that the final
coalgebra of  the functor $2\times\mathit{Id}^A$ is $2^{A^*}$).

It is also interesting to remark the difference between the two
equivalences in the case of partial automata.
%
%It is also interesting to remark what is difference between the two
%equivalences our second motivating example of partial automata.
%
The coalgebraic semantics of PAs~\cite{jan99} is given in terms of pairs of
prefix-closed languages $<V,W>$ where $V$ contains the
words that are accepted (that is, are the label of a path leading to
a final state) and $W$ contains all words that label any path (that
is all that are in $V$ plus the words labeling paths leading to
non-final states). We describe $V$ and $W$  in the
following two examples, for the states $s_0$ and $q_0$:
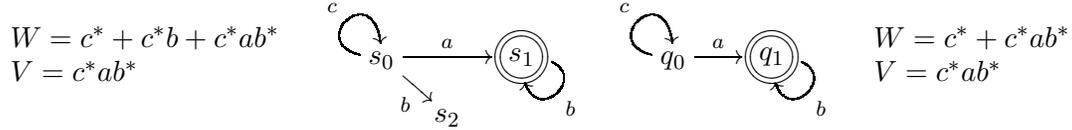
\[
\begin{array}{llll}
\begin{array}{l}
W = c^* + c^*b + c^*ab^*\\
V = c^*ab^*
\end{array}
& \xymatrix@C=0.3cm@R=0.2cm{
s_0 \ar[dr]_-b \ar[rr]^-a \ar@(l,u)^{c} & & *++[o][F=]{s_1}  \ar@(r,d)^{b} \\
&s_2} & \xymatrix@C=0.2cm@R=0.2cm{q_0 \ar[rr]^-a
\ar@(l,u)^{c}&&*++[o][F=]{q_1}\ar@(r,d)^{b}} &
\begin{array}{l}
W = c^* + c^*ab^*\\
V = c^*ab^*
\end{array}
\end{array}
\]
%
%
%Thus, state $s_0$ and $q_0$ would be distinguished by
%$FT$-bisimilarity, but they are equivalent with respect to the monad
%$1+-$ : $s_0\approx_F^{T} q_0$, since they recognize the same
%language.
%
Thus, the states $s_0$ and $q_0$ would be distinguished by
$FT$-equivalence (for $F=2\times Id^A$ and $T=1+-$) but they are
equivalent with respect to the monad $1+-$, $s_0\approx_F^{T} q_0$,
since they accept the same language.

%FILIPPO: I Would comment this
We will show in Section~\ref{sec:bisim_implies_trace} that the
equivalence $\sim_{FT}$ is always contained in $\approx^T_F$.
\subsection{Examples}\label{sec:examples}
In this section we show more examples of applications of the
framework above.
%\vspace*{-.5cm}

\subsubsection{Partial Mealy machines}
A partial Mealy machine is a set of states $X$ together with a
function $t\colon X\to (B\times (1+X))^A$, where $A$ is a set of
inputs and $B$ is a set of output values. We assume that
$B$ has a distinguished element
$\bot \in B$. For each state $x$ and for each input $a$ the automaton
produces an output value and either terminates or continues to a
next state.  Applying the framework above we will be
\emph{totalizing} the automaton, similarly to what happened in the
example of partial automata, by adding an extra state to the state
space which will act as a sink state.  The behaviour of the
totalized automaton is given by the set of causal functions from
$A^\omega$ (infinite sequences of $A$) to $B^\omega$, which we
denote by $\Gamma(A^\omega, B^\omega)$~\cite{jan_mealy}. A function
$f\colon A^\omega \to B^\omega$ is causal if, for $\sigma\in
A^\omega$, the $n$-th value of the output stream $f(\sigma)$ depends
only on the first $n$ values of the input stream $\sigma$.
In the diagram below, we define the final map
$\bb{-} \colon 1+X  \to \Gamma(A^\omega, B^\omega)$:
\[
\xymatrix{ X \ar[dd]_{t}\ar@{->}[r]^-{\kappa_2}&
1+X\ar[ddl]^-{t^\sharp}\ar@{-->}[rrr]^{\bb-}
\ar@{}[ddrrr]|{\small\begin{array}{l}
\bb{\kappa_1(*)} (\sigma) = (\bot, \bot, \ldots) \\
\bb{\kappa_2(x)} (a \,: \, \tau) =  b \, : \, (\bb{z}(\tau))\\
\hspace{1cm} \text{ where } t(x)(a) = <b,z>
\end{array}}
&&&
\Gamma(A^\omega, B^\omega)\ar[dd]\\\\
(B\times (1+X))^A\ar@{-->}[rrrr] &&&& (B\times \Gamma(A^\omega,
B^\omega))^A }
\]
Here $* \in 1$, $x \in X$, $a \in A$, $b \in B$,
$\sigma \in A^\omega$, $z \in 1+X$,
and $a: \tau$ denotes the prefixing of the stream $\tau \in A^\omega$ with the element $a$.

\subsubsection{Structured Moore automata}

In the following examples we look at the functor
\[
F(X) = T(B) \times X^A
\]
for arbitrary sets $A$ and $B$ and an arbitrary monad
$\T = (T, \eta, (-)^\sharp)$. The coalgebras of $F$ represents Moore
automata with outputs in $T(B)$ and inputs in $A$. Since $T(B)$ is a $\T$-algebra, $T(X)^A$ is a $\T$-algebra 
and the product of $\T$-algebras is still a $\T$-algebra, then $FT(X)$ is a $\T$-algebra. 
%Because $FT(X)$ is a $\T$-algebra,
%for any set $X$, 
For this reason, 
the (pair of) functions
$o \colon X \to T(B)$ and $t \colon X \to T(X)^A$
lift to a (pair of) functions
\[
o^\sharp \colon T(X) \to T(B)
\;\;\;\;\;\;
t^\sharp  \colon T(X) \to T(X)^A
\]
The final coalgebra of $F$ is
$T(B)^{A^*}$. We can characterize the final map $\bb{-} \colon T(X) \to T(B)^{A^*}$,
for all $m \in T(X)$, $a \in A$ and $w \in A^*$, by
%\vspace*{-.4cm}
\[
\xymatrix@R=1.5cm{
X \ar[d]_{<o,t>}\ar[r]^{\eta_X} &
T(X)\ar@{}[drrr]|{\small\begin{array}{l}
\bb{m }(\epsilon) = o^\sharp(m)\\
\bb{m}(a \cdot w) = \bb{t^\sharp(m)(a)}(w)\\
\end{array}  }
\ar[dl]_-{<o^\sharp,t^\sharp>}\ar@{-->}[rrr]^{\bb{-}} &&& T(B)^{A^*}
\ar[d]^{<\epsilon, (-)_a>}\\
T(B) \times T(X)^A\ar@{-->}[rrrr] &&&& T(B)\times (T(B)^{A^*})^A
} %
\]
Below we shall look at various concrete instances of this scheme,
for different choices of the monad $T$.

%\vspace*{-.5cm}
\paragraph{\em Moore automata with exceptions}
Let $E$ be an arbitrary set, the elements of which we think of
as exceptions. We consider the \emph{exception monad}
$T(X) = E + X$ which has the function $\eta(x) =
\kappa_2(x)$ as its unit. We define the lifting
$f^\sharp\colon T(X)\to T(Y)$,
for any function $f\colon X \to T(Y)$,
by $f^\sharp = [\mathit{id}, f]$.

An $FT$-coalgebra $<o,t>\colon X\to (E+B) \times (E+X)^A$ will
associate with every state $x$ an output value (either in $B$ or an
exception in $E$) and, for each input $a$, a next state or an
exception. The behaviour of a state $x$, given by $\bb{\eta(x)}$,
will be a formal power series over $A$ with output values in $E +B$;
that is, a function from $A^*$ to $E+B$. The final map
is defined as follows, for all $e \in E$, $x \in X$, $a \in A$, and
$w \in A^*$:
%%\vspace*{-.3cm}
%\[
%\begin{array}{l}
%\bb{\kappa_1(e)} (w) = \kappa_1(e) \ \ \ \ \bb{\kappa_2(s)}
%(\epsilon) =   o(s)\ \ \ \ \bb{\kappa_2(s)} (aw) =
%\bb{t(s)(a)}(w)\text{.}
%\end{array}
%\]
\[
\xymatrix@R=0.75cm{
X \ar[dd]_{<o,t>}\ar@{->}[r]^-{\kappa_2}&
E+X\ar[ddl]_-{<o^\sharp,t^\sharp>}\ar@{-->}[rrr]^{\bb-}
\ar@{}[ddrrr]|{\small\begin{array}{l}
\bb{\kappa_1(e)} (w) = \kappa_1(e) \\
\bb{\kappa_2(x)} (\epsilon) =   o(x)\\
\bb{\kappa_2(x)} (a \cdot w) = \bb{t(x)(a)}(w) \\
\end{array}}
&&&
(E+B)^{A^*}\ar[dd]\\\\
(E+B)\times (E+X)^A\ar@{-->}[rrrr] &&&& (E+B)\times ((E+B)^{A^*})^A
}
\]
\paragraph{\em Moore automata with side effects}
Let $S$ be an arbitrary set of so-called \emph{side-effects}.
We consider the monad $T(X) = (S\times
X)^S$, with unit $\eta$ defined, for all $x \in X$
and $s\in S$, by $\eta(x)(s) = <s,x>$.
We define the lifting $f^\sharp\colon T(X)\to T(Y)$
of a function $f\colon X \to T(Y)$
by $f^\sharp (g) (s) = f(x)(s')$, for any $g \in T(X)$ and $s\in S$,
and with $g(s) =<s',x>$.

Consider an $FT$-coalgebra $<o,t>\colon X\to (B\times S)^S \times
((S\times X)^S)^A$ and let us explain the intuition behind this
type of automaton type.
The set $S\times X$
can be interpreted as the configurations of the automaton, where $S$
contains information about the state of the system and $X$ about the
control of the system.
Using the isomorphism $X \to (S\times
B)^S\cong S\times X \to S\times B$,
we can think of $o\colon X \to (S\times
B)^S$ as a function that for each configuration in $S\times X$ provides
an output in $B$ and the new state of the system in $S$.
The transition function $t\colon
X\to ((S\times X)^S)^A$ gives a new configuration for each input
letter and current configuration, using again the fact that $X\to
((S\times X)^S)^A \cong S\times X \to (S\times X)^A$.
In all of this, a concrete instance of the set of side-effects
could be, for example, the set $S=V^L$ of
functions associating memory locations to values.

The behaviour of a state $x \in X$ will be given by $\bb{\eta(x)}$, where
the final mapping is as follows. For all $g \in (S \times X)^S$, $s \in S$,
$a \in A$ and $w \in A^*$, and
with $g(s) = <s',x>$, we have
\[
\xymatrix@R=0.75cm{
X \ar[dd]_{<o,t>}\ar@{->}[r]^-{\eta}&
(S\times X)^S\ar[ddl]_-{<o^\sharp,t^\sharp>}\ar@{-->}[rrr]^{\bb-}
\ar@{}[ddrrr]|{\small\begin{array}{l}
\bb{g} (\epsilon) (s) = o(x)(s') \\
\bb{g} (a \cdot w) = \bb{\lambda s. t(x)(a)(s')}(w) \\
\end{array}}
&&&
((B\times S)^S)^{A^*}\ar[dd]\\\\
(B\times S)^S \times ((S\times X)^S)^A \ar@{-->}[rrrr] &&&&
(B\times S)^S \times (((B \times S)^S)^{A^*})^A
}
\]

%\paragraph{\em Moore automata with interactive input}
%
%Let $I$ be an arbitrary set of \emph{inputs}. Consider the interactive input monad defined by
%the functor $T(X) =  \mu v. X + (v^I)$, i.e. the set of $I$-branching trees with $X$ labeled leaves.
%It comes together with the natural transformation $\eta_X$ mapping $x \in X$ to the tree $x \in T(x)$
%consisting of only one leaf labeled by $x$. The he lifting $f^\sharp \colon T(X) \to T(Y)$
%of a function $f\colon X\to T(Y)$ given by replacing the leaves $x$ of a tree in $c \in T(X)$ with the
%tree $f(x)$. We consider  $FT$-coalgebras
%\[
%<o,t>\colon X
%\to
%T(B) \times T(X)^A
%\]
%The behaviour of a state $x$ is given by $\bb{\eta(x)} = \bb{x}$, where, for
%every tree $c \in T(X)$, $\bb{c}\colon A^* \to T(B)$, is given by mapping
%$\bb{c}(\epsilon)$ to the tree in $T(B)$ obtained by replacing each leave $x$ of $c$ with the tree $o(x)$;
%and $\bb{c}(aw)$ to the the tree in $T(B)$ obtained by replacing each leave $x$ of $c$ with the tree
%$\bb{t(x)(a)}(w))$.

\paragraph{\em Moore automata with interactive output}

Let $O$ be an arbitrary set of \emph{outputs}. Consider the interactive output monad defined by
the functor $T(X) =  \mu v. X + (O\times v) \cong O^*\times X$ together with the natural transformation
$\eta_X = \lambda x \in X. <\epsilon, x>$, and for which the lifting
$f^\sharp \colon T(X) \to T(Y)$ of a function $f\colon X\to T(Y)$ is given by
$f^\sharp (<w,x>) = <ww',y>  \text{ with } f(x) = <w',y> $.
We consider  $FT$-coalgebras
\[
<o,t>\colon X
\to
(O^*\times B) \times (O^*\times X)^A
\]
For $B=1$, the above coalgebras coincide with \emph{(total) subsequential transducers}~\cite{helle}: 
$o\colon X\to O^*$ is the final output function; $t \colon X \to (O^* \times X)^A$ is the
pairing of the output function and the next state-function.

The behaviour of a state $x$ will be given by 
$\bb{\eta(x)} = \bb{<\epsilon, x>}$, where, for every 
$<w,x>\in O^*\times X$, $\bb{<w,x>}\colon A^* \to O^*$, is given by
\[
\begin{array}{lcl@{\hspace{1.5cm}}lcl}
\bb{<w,x>} (\epsilon) &=& w\cdot o(x) &
\bb{<w,x>} (aw_1) &=& w \cdot (\bb{t(x)(a)}(w_1))
\end{array}
\]

\paragraph{\em Probabilistic Moore automata}
Consider the monad of probability distributions defined,
for any set $X$, by
\[
T(X) =\mathcal D_\omega(X)
\]
Its unit is given by the Dirac distribution, defined for $x,x' \in X$ by
\[
\eta(x)(x') = \,
\begin{cases} 1 & x=x' \\ 0 &\text{otherwise}\end{cases}
\]
The lifting $f^\sharp \colon T(X) \to T(Y)$
of a function $f\colon X\to T(Y)$ is given,
for any distribution $c \in \mathcal D_\omega(X)$
and any $y \in Y$,  by
\[
f^\sharp(c)(y) = \,  \sum_{d\in \mathcal D_\omega(Y)} \left(\sum_{x\in
f^{-1}(d)} c(x)\right)\times
d(y)
\]
We will consider $FT$-coalgebras
\[
<o,t> \colon X \to  \mathcal
D_\omega(B) \times \mathcal
D_\omega(X)^A
\]
More specifically, we take $B=2$ which implies $\mathcal D_\omega(2) \cong
[0,1]$. For this choice of $B$, the above $FT$-coalgebras are precisely
the \emph{(Rabin) probabilistic
automata}~\cite{rabin}.
Each state $x$ has an output value in $o(x) \in [0,1]$ and, for each
input $a$, $t(x)(a)$ is a probability distribution of next states. The
behaviour of a state $x$ is given by $\bb{\eta(x)}\colon A^* \to
[0,1]$, defined below. Intuitively, one can think of $\bb{\eta(x)}$
as a probabilistic language: each word is associated with a value
$p\in[0,1]$.
The final mapping
\[
\xymatrix@R=0.75cm{
X \ar[dd]_{<o,t>}\ar@{->}[r]^-{\eta}&
\mathcal D_\omega(X)
\ar[ddl]_-{<o^\sharp,t^\sharp>}\ar@{-->}[rrr]^{\bb-}
%\ar@{}[ddrrr]|
%{\small
%\begin{array}{l}
%\bb{d} (\epsilon) = \sum\limits_{b\in[0,1]}(\sum\limits_{o(x) = b}
%d(x))\times b \\
%\bb{g} (a \cdot w) = \bb{\lambda s. t(x)(a)(s')}(w) \\
%\end{array}}
&&&
[0,1]^{A^*}\ar[dd]\\\\
[0,1] \times \mathcal D_\omega(X)^A \ar@{-->}[rrrr] &&&&
[0,1] \times ([0,1]^{A^*})^A
}
\]
is given, for any $d \in \mathcal D_\omega(X)$, $x\in X$, $a \in A$,
and $w \in A^*$, by
\[
\begin{array}{lcl}
\bb{d} (\epsilon) &=& \sum\limits_{b\in[0,1]}(\sum\limits_{o(x) = b}
d(x))\times b\\
\bb{d} (aw) &=& \bb{\lambda x' . \sum\limits_{c\in
\mathcal D_\omega(X)} (\sum_{b=t(x)(a)} d(x)) \times  c(x')} (w)
\end{array}
\]
It is worth noting that this exactly captures the semantics of
\cite{rabin}, while the ordinary $\sim_{FT}$ coincides with
\emph{probabilistic bisimilarity} of \cite{LarsenS91}.
Moreover $\approx_F^T$ coincides 
with the trace semantics of probabilistic transition systems defined in \cite{HJS} (see Section 7.2 of \cite{JSS}).

\subsection{Pushdown automata, coalgebraically}\label{pda}
\newcommand\pda{{\sc pda}}

Recursive functions in a computer program lead naturally to a stack of recursive function
calls during the execution of the program. In this section, we provide a coalgebraic model
of automata equipped with a stack memory. A \emph{pushdown machine} is a tuple $(Q,A,B,\delta)$,
where $Q$ is set of control locations (states), $A$ is a set of input symbols, $B$ is a set of stack symbols,
and $\delta$ is finite subset of $Q \times A \times B \times Q \times B^*$, called the set of
transition rules. Note that we do not insist on the sets $Q$, $A$ and $B$ to be finite and
consider only \emph{realtime} pushdown machines, i.e. without internal
transitions (also called $\epsilon$-transitions)~\cite{HU79}. A {\em configuration} $k$
of a pushdown machine  is a pair $<q,\beta>$ denoting the current control state
$q \in Q$ and the current content of the stack $\beta \in B^*$. In denoting the
stack as a string of stack symbols we assume that the topmost symbol is written
first. There is a transition $<q,b\beta> \xrightarrow {a} <q',\alpha\beta>$ if $<q',\alpha> \in \delta(q,a,b)$.
A convenient notation is to introduce for any string $w \in A^*$ the transition relation on configurations
as the least relation such that
\begin{enumerate}[(1)]
\item $k \xrightarrow {\epsilon} k$
\item $k \xrightarrow {aw} k'$ if and only if $k \xrightarrow {a} k''$ and $k'' \xrightarrow {w} k'$.
\end{enumerate}
A \emph{pushdown automaton} (\pda) is a pushdown machine together with an initial configuration $k_0$ and a
set $K$ of accepting configurations. The sets of accepting configurations usually considered are
(1)  the set $F \times B^*$, where $F \subseteq Q$ is called the set of accepting states, or
(2) $Q \times \{ \epsilon \}$, but also (3) $F \times \{ \epsilon \}$ for $F \subseteq Q$, or
(4) $Q \times B'B^*$ for $B'$ a subset of $B$.
A word $w \in A^*$ is said to be accepted by  a \pda\ $(Q,A,B,\delta,k_0,K)$ if
$k_0  \xrightarrow {w} k$  for some $k \in K$. A \pda\ with accepting configurations as in (1)
is said to be with accepting states, whereas, when they are as in (2) then the \pda\ is said to be 
accepting by empty stack. They both accept exactly proper context free languages (i.e. context 
free languages without the empty word)~\cite{ABB97}. 

Computations in a pushdown machine are generally non-deterministic and can cause a change in
the control state of the automaton as well as in its stack. For this reason we will model
the effects of the computations by means of the so-called \emph{non-deterministic side-effect}
monad~\cite{BHM00}. For a set of states $S$, let $T$ be the functor $\pow(- \times S)^S$. It is a
monad when equipped with the unit $\eta_X\colon X \to T(X)$, defined by $\eta(x)(s) = \{ <x,s> \}$,
and the multiplication $\mu_X\colon T(T(X)) \to T(X)$ given by
\[
\mu_X(k)(s) = \bigcup_{<c,s'> \in k(s)} c(s')
\]
Note that, for a function $f \colon X \to T(Y)$, the extension $f^\sharp\colon T(X) \to T(Y)$ is defined by
\[
f^\sharp(c)(s) = \bigcup_{<x',s'> \in c(s)} f(x')(s') \,.
\]
Examples of algebras for this monad are $T(1) = \pow(S)^S$ and $2^S$. The latter can in fact
be obtained as a quotient of the former by equating those functions $k_1,k_2\colon S \to \pow(S)$
such that for all $s \in S$, $k_1(s) = \emptyset$ if and only if $k_2(s) = \emptyset$.

Every pushdown machine $(Q,A,B,\delta)$ together with a set of accepting configurations $K$
induces a function $<o,t>\colon Q \to FTQ$ where $F$ is  the functor $2^{B^*} \times id^A$ and
$T$ is the monad defined above specialized for $S=B^*$ (intuitively, side effects in a pushdown machine
are changes in its stack). The functions $o\colon Q \to 2^{B^*}$ and $t\colon Q \to \pow(Q\times B^{*})^{{B^*}^A}$ are defined as
\[
\begin{array}{lcl}
o(q)(\beta) &=& 1  \;\; \mbox{if and only if $<q,\beta> \in K$}
\\
t(q)(a)(\epsilon) &=&  \emptyset \\ 
t(q)(a)(b\beta) &=& \{ <q',\alpha\beta> \mid <q',\alpha> \in \delta(q,a,b) \}
\end{array}
\]
The transition function $t$ describes the steps between \pda\ configurations and it is specified in
terms of the transition instructions $\delta$ of the original machine.

From the above is clear that not every function $<o,t>\colon Q \to FTQ$ defines a pushdown
machine with accepting configurations, as, for example, $t(q)$ may depend on the whole stack $\beta$ and not just on
the top element $b$. Therefore we restrict our attention to consider functions
$<o,t>\colon Q \to FTQ$ such that
\begin{enumerate}[(1)]
\item $t(q)(a)(\epsilon) =  \emptyset$
\item $t(q)(a)(b\beta) = \{ <q', \alpha\beta> \mid <q',\alpha> \in t(q)(a)(b) \}$,
\end{enumerate}
Every $<o,t>$ satisfying (1) and (2) above defines the pushdown machine $(Q,A,B,\delta)$ with
$\delta(q,a,b) = t(q)(a)(b)$ and with accepting configuration $K = \{ <q,\beta> | o(q)(\beta) = 1 \}$.
The first condition is asserting that a machine is in a deadlock configuration when the stack is empty,
while the last condition ensures that transition steps depend only on the control state
and the top element of the stack. For this reason we will write $q \xrightarrow {a,b|\alpha}  q'$
for $<q',\alpha \beta> \in t(q)(a)(b)$ indicating that the pushdown machine in the state $q$ by reading an
input symbol $a$ and popping $b$ off the stack, can move to a control state $q'$ pushing the string
$\alpha \in B^*$ on the current stack (here denoted by $\beta$).

Similarly to what we have shown in the examples of structured Moore automata, for every function
$<o,t>\colon Q \to FTQ$ there is a unique $F$-coalgebra map $\bb{-}\colon T(Q) \to 2^{{B^*}^{A^*}}$,
which is also a $T$-algebra homomorphism. It is defined for all $c\in \pow(Q \times B^*)^{B^*}$ and $\beta \in B^*$ as
\[
\xymatrix@R=0.75cm@C=.6cm{
Q \ar[dd]_{<o,t>}\ar@{->}[r]^(0.36){\eta}& \pow(Q \times B^*)^{B^*}\ar[ddl]^-{<o^\sharp, t^\sharp>}\ar@{-->}[r]^(0.6){\bb-}
&
{2^{B^*}}^{A^*}\ar[dd]\\\\
2^{B^*}\times { \pow(Q \times B^*)^{B^*}}^A \ar@{-->}[rr] &&2^{B^*}\times {{2^{B^*}}^{A^*}}^A
}\raisebox{-1cm}{$
\begin{array}{l}
\bb{\eta(q)}(\epsilon) = o(q)\\
\bb{\eta(q)}(aw) = \bb{\lambda \beta.t(q)(a)(\beta)}(w)\\
\bb{c}(\beta)=\bigcup\limits_{<q,\alpha> \in c(\beta)} \bb{\eta(q)}(\alpha)\,.
\end{array}$}
\]
We then have that a word $w \in A^*$ is \emph{accepted} by the \pda\ $(Q,A,B,\delta,k_0,K)$ with $k_0 = <q,\beta>$ if
and only if $\bb{\eta(q)}(w)(\beta) = 1$. 

\smallskip

The above definition implies that for a given word $w \in A^*$ we can decide
if it is accepted by $<o,t>\colon Q \to FTQ$ from an initial configuration $k_0 = <q,\beta>$
in exactly $|w|$ steps (assuming there is a procedure to decide whether $o(q)(\beta) = 1$). As a consequence, we cannot
use structured Moore automata to model Turing machines, for which the halting problem is undecidable: in general terms, 
for Turing machines, we would need internal transitions that do not consume input symbols.

%
% OLD
%
%As a consequence, we cannot
%model more powerful machines, such as Turing machines, for which the halting problem is known to be undecidable, by
%transforming them into structured Moore automata, because computations, in general, will not be realtime (i.e. using
%internal transitions without consuming input symbols).
\smallskip

We conclude with an example of our construction using a pushdown machine with control states $Q=\{q_0, q_1\}$, over an
input alphabet $A = \{ a,b \}$ and using stack symbols $B = \{ x,s\}$. The transitions rules $\delta$ are given below:
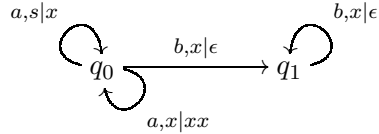
\begin{center}\[
\xymatrix{
q_0 \ar@(l,u)^{a,s|x} \ar@(r,d)^{a,x|xx} \ar[rr]^-{b,x|\epsilon} & & q_1 \ar@(r,u)_{b,x|\epsilon}
}
\]
\end{center}
We take $K=\{<q_0,\epsilon>,\, <q_1,\epsilon>\}$, meaning that $o(q_0)(\epsilon) = 1$, $o(q_1)(\epsilon) = 1$ 
and $o(q_i)(\beta) = 0$ in all other cases. By considering $k_0 = <q_0,s>$ as
initial configuration, we then have
\[
\bb{\eta(q_0)}(\epsilon)(s) = o(q_0)(s) = 0
\]
meaning that the empty word is not accepted by the \pda\ $(Q,A,B,\delta,k_0, K)$. However, the word $ab$ is accepted:
\[
\begin{array}{lcl}
\bb{\eta(q_0)}(ab)(s) & = & \bb{\lambda \beta.t(q_0)(a)(\beta)}(b)(s)\\
                      & = & \bigcup\limits_{<p,\beta> \in t(q_0)(a)(s)} \bb{\eta(p)}(b)(\beta)\\
                      & = & \bb{\eta(q_1)}(b)(x)\\
                      & = &  \bb{\lambda \beta.t(q_1)(b)(\beta)}(\epsilon)(x)\\
                      & = &  \bigcup\limits_{<p,\beta> \in t(q_1)(b)(x)} \bb{\eta(p)}(\epsilon)(\beta)\\
                      & = &  \bb{\eta(q_1)}(\epsilon)(\epsilon)\\
                      & = &  o(q_1)(\epsilon)\\
                      & = &  1 \,.
\end{array}
\]
In fact, the language accepted by the above pushdown automaton is $\{ a^nb^n \mid n \geq 1 \}$. 
The structured states $c_i\in TQ$, their transitions and their outputs of (part of) the associated Moore automaton are given in Figure \ref{fig:pda1}.
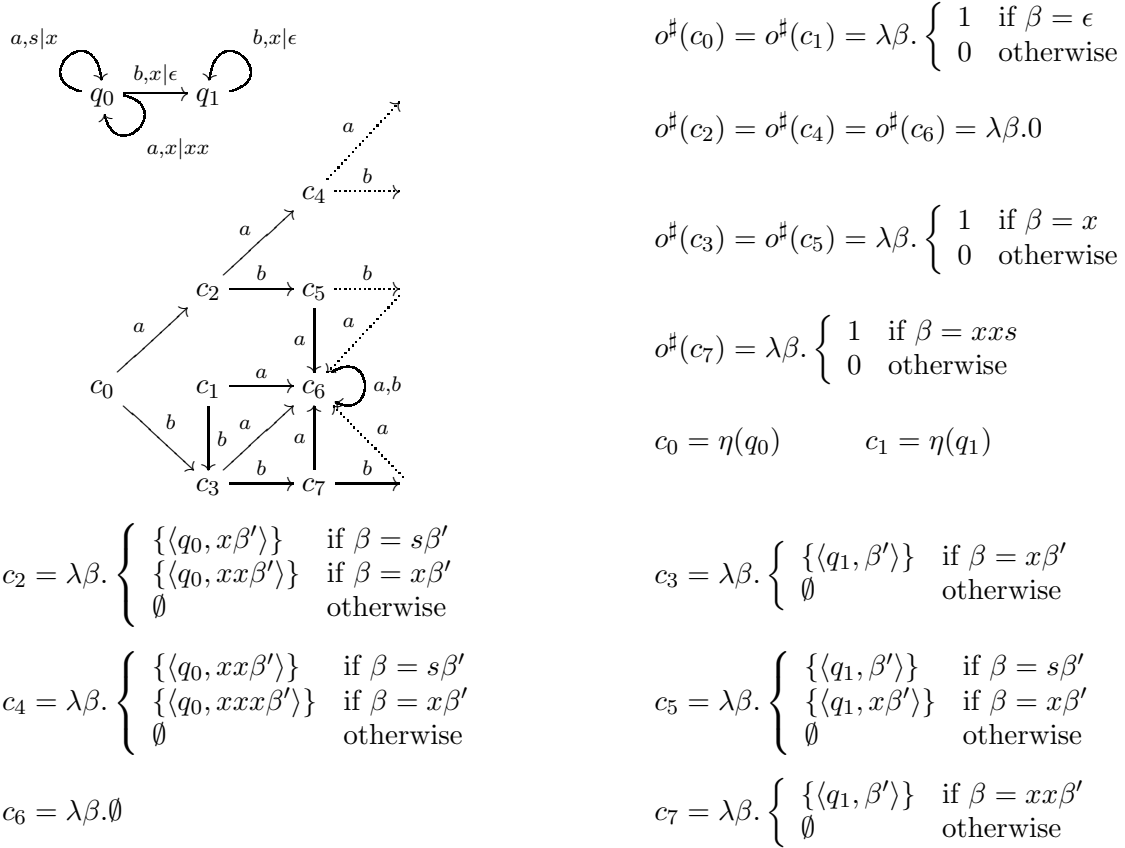
\begin{figure}
\begin{tabular}{lcl}
\multirow{8}{*}{$\xymatrix{
q_0 \ar@(l,u)^{a,s|x} \ar@(r,d)^{a,x|xx} \ar[r]^-{b,x|\epsilon}  & q_1 \ar@(r,u)_{b,x|\epsilon} & &\\
 &                              & c_4 \ar@{.>}[ru]^-{a} \ar@{.>}[r]^-{b} & \\
 & c_2 \ar[ru]^-{a} \ar[r]^-{b} & c_5 \ar[d]_-{a}\ar@{.>}[r]^-{b} & \ar@{.>}[ld]_-{a}\\
c_0 \ar[ru]^-{a}\ar[rd]^-{b} & c_1 \ar[r]^-{a} \ar[d]^-{b} & c_6  \ar@(ur,dr)^{a,b}& & & \\
    & c_3  \ar[ru]^-{a} \ar[r]^-{b}  & c_7 \ar[u]^-{a} \ar[r]^-{b}  & \ar@{.>}[lu]_-{a}& \\
}$}
& \;\;\; &
$o^\sharp(c_0) = o^\sharp(c_1) = \lambda \beta.\left\{
\begin{array}{ll}
1 & \mbox{if $\beta = \epsilon$}\\
0 & \mbox{otherwise}
\end{array}
\right.$
\\[2em]
& & $o^\sharp(c_2) = o^\sharp(c_4) = o^\sharp(c_6) = \lambda \beta.0$
\\[2em]
& & $o^\sharp(c_3) = o^\sharp(c_5) =\lambda \beta.\left\{
\begin{array}{ll}
1 & \mbox{if $\beta = x$}\\
0 & \mbox{otherwise}
\end{array}
\right.$
\\[2em]
& & $o^\sharp(c_7) = \lambda \beta.\left\{
\begin{array}{ll}
1 & \mbox{if $\beta = xxs$}\\
0 & \mbox{otherwise}
\end{array}
\right.$
\\[2em]
& & $c_0 = \eta(q_0)$ \;\;\;\;\;\;\;\;  $c_1 = \eta(q_1)$
\\[2.2em]
  $c_2 = \lambda \beta.\left\{
\begin{array}{ll}
\{<q_0,x\beta'>\} & \mbox{if $\beta = s\beta'$}\\
\{<q_0,xx\beta'>\}& \mbox{if $\beta = x\beta'$}\\
\emptyset & \mbox{otherwise}
\end{array}
\right.$
& & $c_3 = \lambda \beta.\left\{
\begin{array}{ll}
\{<q_1,\beta'> \} & \mbox{if $\beta = x\beta'$}\\
\emptyset & \mbox{otherwise}
\end{array}
\right.$
\\[2em]
$c_4 = \lambda \beta.\left\{
\begin{array}{ll}
\{<q_0,xx\beta'>\}& \mbox{if $\beta = s\beta'$}\\
\{<q_0,xxx\beta'>\} & \mbox{if $\beta = x\beta'$}\\
\emptyset & \mbox{otherwise}
\end{array}
\right.$
& &
$c_5 = \lambda \beta.\left\{
\begin{array}{ll}
\{<q_1,\beta'>\} & \mbox{if $\beta = s\beta'$}\\
\{<q_1,x\beta'>\} & \mbox{if $\beta = x\beta'$}\\
\emptyset & \mbox{otherwise}
\end{array}
\right.$
\\[2em]
$c_6 = \lambda \beta.\emptyset$
& &
$c_7 = \lambda \beta.\left\{
\begin{array}{ll}
\{<q_1,\beta'>\} & \mbox{if $\beta = xx\beta'$}\\
\emptyset & \mbox{otherwise}
\end{array}
\right.$
\end{tabular}\caption{
The structured states $c_i\in TQ$, their transitions and their output of (part of) the Moore automaton associated to the 
\pda\ $(Q,A,B,\delta,k_0,K)$ where $Q=\{q_0,q_1\}$, $A=\{a,b\}$, $B=\{x,s\}$, $\delta$ is depicted on the left top, $k_0 = <q_0,s>$ and $K=\{<q_0,\epsilon>,\, <q_1,\epsilon>\}$.}\label{fig:pda1}
\end{figure}

\medskip

\emph{Context-free grammars} generating proper languages (i.e. not containing the empty word $\epsilon$)
are equivalent to realtime \pda's~\cite{Cho62,Eve63,Sch63}. Given an input alphabet $A$, and a set of
variables $B$, let  $G = (A, B, s, P)$ be a context-free grammar in Greibach normal form~\cite{Gre67},
i.e. with productions in $P$ of the form  $b \to a\alpha$ with $b \in B$, $a \in A$ and $\alpha \in B^*$.
We can construct a function $<o,t> \colon 1 \to FT1$ (where $1=\{*\}$) by setting
\[
o(*)(\beta) = 1\; \mbox{if and only if}\; \beta = \epsilon
\;\;\;\;\; \mbox{ and }\;\;\;\;\;
t(*)(a)(b\beta) = \{ <*, \alpha\beta> \mid b\to a\alpha \in P \} \,.
\]
Clearly this function satisfies conditions $(1)$ and $(2)$ above, and thus, together with the initial
configuration $<*,s>$ defines a \pda. Furthermore, $\bb{\eta(*)}(w)(s) = 1$ if and only if
there exists a derivation for $w \in A^*$ in the grammar $G$.

As an example, let us consider the grammar $(\{a,b\},\{s,x\}, s, P)$ with productions
$P = \{ s \to asx, s \to ax, x \to b \}$  generating the language $\{ a^nb^n \mid n \geq 1 \}$.
The associated coalgebra $<o,t>\colon 1 \to FT1$ is given by
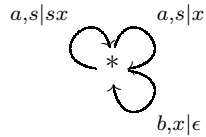
\[
\xymatrix{
\mbox{*} \ar@(l,u)^{a,s|sx} \ar@(u,r)^{a,s|x} \ar@(r,d)^-{b,x|\epsilon}
& &
\mbox{with $o(*)(\beta) = 1$ iff $\beta = \epsilon$}
}
\]
Even if the language accepted by the above \pda\, is the same as the one accepted by
the \pda\, in the previous example (i.e., $\bb{\eta(*)}(w)(s) = \bb{\eta(q_0)}(w)(s)$ for all $w\in A^*$), the two associated 
Moore automaton are not in $\approx_F^T$ (that is $\bb{\eta(*)} \neq \bb{\eta(q_0)}$).
In fact, the Moore automaton associated to the above coalgebra (see below) accepts the string 
$abab$ when starting from the configuration $<*,ss>$, while the one in the previous example does not (in symbols, 
$\bb{\eta(*)}(abab)(ss)=1$ while $\bb{\eta(q_0)}(abab)(ss)=0$).

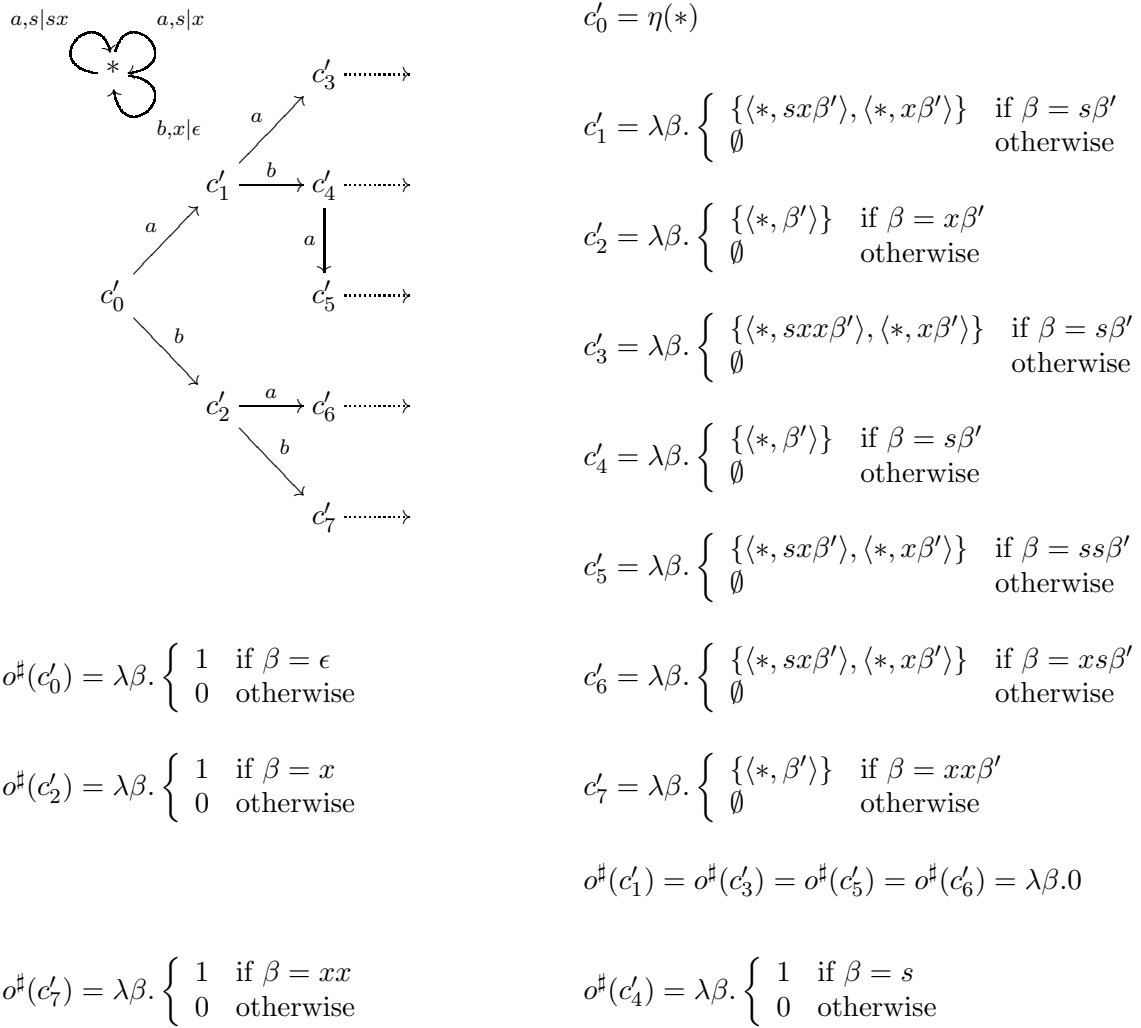
\begin{figure}
\begin{tabular}{lcl}
\multirow{8}{*}{$\xymatrix{
                 \mbox{*} \ar@(l,u)^{a,s|sx} \ar@(u,r)^{a,s|x} \ar@(r,d)^-{b,x|\epsilon} &     & c_3' \ar@{.>}[r] & \\
                 & c_1' \ar[ru]^-{a} \ar[r]^-{b} & c_4' \ar[d]_-{a}\ar@{.>}[r] & \\
c_0' \ar[ru]^-{a}\ar[rd]^-{b} &  & c_5' \ar@{.>}[r]& \\
    & c_2'  \ar[r]^-{a} \ar[rd]^-{b}  & c_6' \ar@{.>}[r]  & & \\
    &                                 & c_7' \ar@{.>}[r] & & }$}
& \;\;\; &
$c_0' = \eta(*)$
\\[2em]
& & $c_1' = \lambda \beta.\left\{
\begin{array}{ll}
\{<*,sx\beta'>,<*,x\beta'>\} & \mbox{if $\beta = s\beta'$}\\
\emptyset & \mbox{otherwise}
\end{array}
\right.$
\\[2em]
& & $c_2' = \lambda \beta.\left\{
\begin{array}{ll}
\{<*,\beta'> \} & \mbox{if $\beta = x\beta'$}\\
\emptyset & \mbox{otherwise}
\end{array}
\right.$
\\[2em]
& & $c_3' = \lambda \beta.\left\{
\begin{array}{ll}
\{<*,sxx\beta'>,<*,x\beta'>\} & \mbox{if $\beta = s\beta'$}\\
\emptyset & \mbox{otherwise}
\end{array}
\right.$
\\[2em]
& & $c_4' = \lambda \beta.\left\{
\begin{array}{ll}
\{<*,\beta'> \} & \mbox{if $\beta = s\beta'$}\\
\emptyset & \mbox{otherwise}
\end{array}
\right.$
\\[2em]
& & $c_5' = \lambda \beta.\left\{
\begin{array}{ll}
\{<*,sx\beta'>,<*,x\beta'>\} & \mbox{if $\beta = ss\beta'$}\\
\emptyset & \mbox{otherwise}
\end{array}
\right.$
\\[2em]
$o^\sharp(c_0') = \lambda \beta.\left\{
\begin{array}{ll}
1 & \mbox{if $\beta = \epsilon$}\\
0 & \mbox{otherwise}
\end{array}
\right.$& & $c_6' = \lambda \beta.\left\{
\begin{array}{ll}
\{<*,sx\beta'>,<*,x\beta'>\} & \mbox{if $\beta = xs\beta'$}\\
\emptyset & \mbox{otherwise}
\end{array}
\right.$
\\[2em]
$o^\sharp(c_2') = \lambda \beta.\left\{
\begin{array}{ll}
1 & \mbox{if $\beta = x$}\\
0 & \mbox{otherwise}
\end{array}
\right.$& & $c_7' = \lambda \beta.\left\{
\begin{array}{ll}
\{<*,\beta'>\} & \mbox{if $\beta = xx\beta'$}\\
\emptyset & \mbox{otherwise}
\end{array}
\right.$
\\[2em]
& &
$o^\sharp(c_1') = o^\sharp(c_3') = o^\sharp(c_5') =  o^\sharp(c_6') = \lambda \beta.0$
\\[2em]
$o^\sharp(c_7') = \lambda \beta.\left\{
\begin{array}{ll}
1 & \mbox{if $\beta = xx$}\\
0 & \mbox{otherwise}
\end{array}
\right.$& &
$o^\sharp(c_4') = \lambda \beta.\left\{
\begin{array}{ll}
1 & \mbox{if $\beta = s$}\\
0 & \mbox{otherwise}
\end{array}
\right.$
\end{tabular}\caption{The structured states $c_i\in TQ$, 
their transitions and their output of (part of) the Moore automaton associated to the \pda\ $(Q,A,B,\delta,k_0,K)$ where $Q=\{*\}$, $A=\{a,b\}$, 
$B=\{x,s\}$, $\delta$ is depicted on the left top, $k_0 = <*,s>$ and $K=\{<*,\epsilon>\}$.}
\end{figure}

\medskip

The above characterization of context free languages over an alphabet $A$ is different and complementary
to the coalgebraic account of context-free languages presented in~\cite{WBR11}. The latter, in fact,
uses the functor $D(X) = 2 \times X^A$ for deterministic automata (instead of the Moore automata
with output in $2^{B^*}$ above, for $B$ a set of variables), and the idempotent semiring
monad $T(X) = \pow((X+A)^*)$ (instead of our side effect monad) to study different
but equivalent ways to present context-free languages: using grammars, behavioural differential
equations and generalized regular expressions in which the Kleene star is replaced by a
unique fixed point operator.

\section{Coalgebras and $\T$-Algebras}\label{sec:bisim_implies_trace}
In the previous section we presented a framework, parameterized by a
functor $F$ and a monad $\T$, in which systems of type $FT$ (that
is, $FT$-coalgebras) can be studied using a novel equivalence
$\approx^T_F$ instead of the classical $\sim_{FT}$. The only
requirement we imposed was that $FT(X)$ has to be a $\T$-algebra.

In this section, we will present functors $F$ for which the requirement of $FT(X)$ being a
$\T$-algebra is guaranteed because they can be {\em lifted} to a functor $F^*$ on $\T$-algebra.
For these functors, the equivalence $\approx^T_F$ coincides with $\sim_{F^*}$. In other words,
working on  $FT$-coalgebras in $\Set$ under the novel $\approx^T_F$ equivalence is the same
as working on $F^*$-coalgebras on $\T$-algebras under the ordinary
$\sim_{F^*}$ equivalence. Next, we will prove that for this class of functors and an arbitrary monad $\T$ the equivalence
$\sim_{FT}$ is contained in $\approx^T_F$. Instantiating this result for our first motivating
example of non-deterministic automata will yield the well known fact that bisimilarity implies
trace equivalence.

Let $\T$ be a monad. An endofunctor $F^*\colon\Set^\T \to \Set^\T$ is
said to be the $\T$-\emph{algebra lifting} of a functor $F\colon\Set
\to \Set$ if the following square commutes\footnote{This is
equivalent to the existence of a distributive law $\lambda\colon TF
\Rightarrow FT$ ~\cite{Joh75}.}:
\[
\xymatrix@C=1.75cm@R=1.5cm{
\Set^\T \ar[d]_{U^\T}\ar[r]^{F^*} & \Set^\T \ar[d]^{U^\T}\\
\Set \ar[r]_{F} & \Set
}
\]
If the functor $F$ has a $\T$-algebra
lifting $F^*$ then $FT(X)$ is the carrier of the algebra $F^*(T(X),\mu)$.  Functors that
have a $\T$-algebra lifting are given, for example, by those endo\-functors on $\Set$
constructed inductively by the following grammar
\[
F ::=  \mathit{Id}  \mid B \mid F \times F \mid F^A \mid TG
\]
where $A$ is an arbitrary set, $B$ is the constant functor mapping
every set $X$ to the carrier of a $\T$-algebra $(B,h)$, and $G$ is
an arbitrary functor. Since the forgetful functor $U^\T\colon
\Set^\T \rightarrow \Set$ creates and preserves limits, both $F_1
\times F_2$ and $F^A$ have  a $\T$-algebra lifting if $F$, $F_1$,
and $F_2$ have. Finally, $TG$ has a $\T$-algebra lifting for every
endofunctor $G$ given by the assignment $(X,h) \mapsto
(TGX,\mu_{GX})$. Note that we do not allow taking coproducts in the
above grammar, because coproducts of $\T$-algebras are not preserved
in general by the forgetful functor $U^\T$. Instead, one could
resort to extending the grammar with the carrier of the coproduct
taken directly in $\Set^\T$. For instance, if $\T$ is the (finite)
powerset monad, then we could extend the above grammar with the
functor $F_1 \oplus F_2 = F_1 + F_2 + \{\top,\bot\}$.

All the functors of the examples in Sections \ref{sec:motiv} and \ref{sec:general}, as well as those in Section \ref{secbbat}, 
can be generated by the above grammar and, therefore, they have a $\T$-algebra lifting.

\medskip

Now, let $F$ be a functor with a $\T$-algebra lifting and for which a final coalgebra $\Omega$ exists. If
$\Omega$ can be constructed as the limit of the final sequence (for example assuming the functor accessible~\cite{Ada74}),
then, because the forgetful functor $U^\T\colon\Set^\T \to \Set$ preserves and creates limits, $\Omega$ is the carrier
of a $\T$-algebra, and it is the final coalgebra of the lifted functor $F^*$. Further,
for any $FT$-coalgebra $f\colon X \to FT(X)$, the unique $F$-coalgebra homomorphism $\bb{-}$
as in diagram (\ref{F-final}) is a $T$-algebra homomorphism between $T(X)$ and $\Omega$.
Conversely, the carrier of the final $F^*$-coalgebra (in $\Set^\T$) is the final $F$-coalgebra
(in $\Set$). %Note that for the above result to hold in an arbitrary category, we have to assume
%the existence of enough limits so to guarantee the construction of the final coalgebra
%of $F$ as limit of the final sequence.

Intuitively, the above means that for an accessible functor $F$ with a $\T$-algebra
lifting $F^*$, $F^*$-equivalence in $\Set^\T$ coincides with $F$-equivalence
with respect to $\T$ in $\Set$. The latter equivalence is coarser than the
$FT$-equivalence in $\Set$, as stated in the following theorem.
\begin{thm}\label{thm:main}
Let $\T$ be a monad. If $F$ is an endofunctor on $\Set$ for which a final coalgebra exists and with a $\T$-algebra lifting,
then $\sim_{FT}$ implies $\approx_F^T$.
\end{thm}
\begin{proof}
We first show that there exists a functor from the category of
$FT$-coalgebras to the category of $F$-coalgebras.

This functor maps each $FT$-coalgebra $(X,f)$ into the $F$-coalgebra
$(T(X),f^{\sharp})$ and each $FT$-homomorphism $h \colon (X,f) \to (Y,g)$
into the $F$-homomorphism $T(h) \colon (T(X),f^{\sharp}) \to
(T(Y),g^{\sharp})$.
In order to prove that this is a functor we just have to show that
$T(h)$ is an $F$-homomorphism (i.e., the backward face of the
following diagram commutes).

$$\xymatrix{ & T(X) \ar@{}[dddl]^{f^{\sharp}} \ar[dddl]|(0.36){\hole} \ar[rr]^{T(h)}& & T(Y) \ar[dddl]^{g^{\sharp}}\\
X \ar[rr]^(0.6)h \ar[dd]_f \ar[ur]^{\eta_X} & & Y \ar[dd]^g \ar[ur]^{\eta_Y} \\
\\
FT(X) \ar[rr]_{FT(h)}& & FT(Y)} $$

%\qquad \xymatrix@R=10pt@C=10pt{ & T(X) \ar[dddl]^{f^{\sharp}}
%\ar[rr]^{T(h)}& & T(Y) \ar[dddl]^{g^{\sharp}}  \\ \\ \\ FT(X)
%\ar[rr]_{FT(h)}& & FT(Y)}

%Note that the rightmost diagram coincides with the the leftmost
%diagram.
Note that the front face of the above diagram commutes because $h$ is
an $FT$-homomorphism. Also the top face commutes because $\eta$ is
a natural transformation. Thus $$FT(h)\circ f^{\sharp}\circ \eta_X = FT(h)\circ f = g \circ h$$
 and also  $$g^{\sharp} \circ T(h)\circ \eta_X = g^{\sharp} \circ \eta_Y \circ h = g \circ h\text{.}$$
Since $\eta$ is the
unit of the adjunction, then there exists a unique
$j^{\sharp}\colon T(X)\to FT(Y)$ in $\Set^{\T}$ such that $g \circ
h=j^{\sharp}\circ \eta_X$. Since both $FT(h)\circ f^{\sharp}$ and
$g^{\sharp} \circ T(h)$ are (by construction) morphisms in
$\Set^{\T}$, then $FT(h)\circ f^{\sharp} = g^{\sharp} \circ T(h)$.

Let $(X, f)$ and $(Y,g)$ be two $FT$-coalgebras and $\bb{-}_X$ and $\bb{-}_Y$ their morphisms into the final $FT$-coalgebra $(\Omega, \omega)$.
Let $(T(X),f^{\sharp})$, $(T(Y),g^{\sharp})$ and $(T(\Omega), \omega^{\sharp})$ be the corresponding $F$-coalgebras and $\bb{-}_{TX}$, $\bb{-}_{TY}$ and $\bb{-}_{T\Omega}$
their morphisms into the final $F$-coalgebra $(\Omega', \omega')$.

Since $T(\bb{-}_X) \colon (T(X),f^{\sharp}) \to (T(\Omega), \omega^{\sharp})$ is an $F$-homomorphism, then by uniqueness,
$\bb{-}_{TX} =\bb{-}_{T\Omega} \circ T(\bb{-}_X)$.

\begin{center}
$$\xymatrix{ & T(X)\ar@(ur,ul)[rrrr]^{\bb{-}_{TX}} \ar@{}[dddl]^{f^{\sharp}} \ar[dddl]|(0.36){\hole} \ar[rr]^{T(\bb{-}_X)}& & T(\Omega) \ar[dddl]^{\omega^{\sharp}} \ar[rr]^{\bb{-}_{T\Omega}}& & \Omega' \ar[dddl]^{\omega'} \\
X \ar[rr]^(0.6){\bb{-}_X} \ar[dd]_f \ar[ur]^{\eta_X} & & \Omega \ar[dd]^{\omega} \ar[ur]^{\eta_{\Omega}} \\
\\
FT(X) \ar[rr]_{FT(\bb{-}_X)} \ar@(dr,dl)[rrrr]_{F(\bb{-}_{TX})} & & FT(\Omega) \ar[rr]_{F(\bb{-}_{T\Omega})} & & F(\Omega') }$$
\end{center}

With the same proof, we obtain $\bb{-}_{TY} =\bb{-}_{T\Omega} \circ T(\bb{-}_Y)$. 

Recall that for all $x\in X$ and $y\in Y$, by definition,
$x\sim_{FT} y$ iff $\bb{x}_X = \bb{y}_Y$ and $x\approx_F^T y$ iff $\bb{\eta_X(x)}_{TX} = \bb{\eta_Y(y)}_{TY}$.

Suppose that $\bb{x}_X = \bb{y}_Y$. Then, $ T(\bb{\eta_X (x)}_X) = \eta_{\Omega} \circ \bb{x}_X = \eta_{\Omega} \circ \bb{y}_Y = T(\bb{\eta_Y (y)}_Y)$ and, finally,
$\bb{\eta_X(x)}_{TX} = \bb{-}_{T\Omega} \circ T(\bb{\eta_X(x)}_X) = \bb{-}_{T\Omega} \circ T(\bb{\eta_Y(y)}_Y) = \bb{\eta_Y(y)}_{TY}$.
\end{proof}
%
%
%The proof of this theorem (presented in \cite{TR}) relies on the
%fact that for every monad $\T$ and functor $F$ with a $\T$-algebra
%lifting, if $h\colon(X,f) \rightarrow (Y,g)$ is an $FT$-coalgebra
%homomorphism, then $(\eta_Y \circ h)^\sharp \colon (T(X), f^\sharp)
%\rightarrow (T(Y),g^\sharp)$ is an $F$-coalgebra homomorphism.

The above theorem  instantiates to the well-known facts: for NDA, where $F(X) = 2 \times
X^A$ and $T=\pow$,  that bisimilarity implies language
equivalence; for partial automata,  where $F(X) = 2 \times
X^A$ and $T=1+-$, that equivalence of pairs of languages, consisting
of defined paths and accepted words, implies
equivalence of accepted words; for probabilistic automata, where $F(X)
= [0,1] \times X^A$ and $T=\mathcal D_\omega$, that probabilistic
bisimilarity implies probabilistic/weighted language equivalence. Note that, in general, the
above inclusion is strict.

\begin{remark}
Let $(X,f)$ be an $FT$-coalgebra for a monad $\T$ and a functor $F$. If
$\eta\colon \mathit{id} \Rightarrow T$ is pointwise injective, then
$\sim_{FT}$ on the
$FT$-coalgebra $(X,f)$ coincides with $\sim_{TFT}$ on the extended
$TFT$-coalgebra
$(X, \eta_{FT(X)} \circ f)$~\cite{Rutten00,bartels}. If moreover $F$ has a
$\T$-algebra
lifting then, by the above theorem (on the extended $TFT$-coalgebra), $\sim_{TFT}$  implies
$\approx_{TF}^T$. Combining
the two implications, it follows that hat $\sim_{FT}$ on the
$FT$-coalgebra $(X,f)$ implies
$\approx_{TF}^T$ on the extended  $TFT$-coalgebra $(X, \eta_{FT(X)}
\circ f)$. Finally,
under the assumption that $F$ has a $\T$-algebra lifting, we also have
that $\approx_F^T$
the $FT$-coalgebra $(X,f)$ implies $\approx_{TF}^T$ on the extended
$TFT$-coalgebra
$(X, \eta_{FT(X)} \circ f)$. This yields the following hierarchy of equivalences.
\[
\xymatrix{
& \approx_{TF}^T\\
 & & & \approx_F^T\ar@{-}[ull]_\supseteq\\
\sim_{TFT}
\ar@{-}[uur]^{\subseteq} \ar@{-}[rr]_{=} & &
\sim_{FT}\ar@{-}[ur]_\subseteq
}
\]
\end{remark}

%{\rotatebox[origin=c]{90}{$\subseteq$}}

\section{Beyond Bisimilarity and Traces}\label{secbbat}
The operational semantics of interactive systems is
usually specified by labeled transition systems (LTS's). The denotational semantics is given in terms of behavioural
equivalences, which depend the amount of branching structure
considered. Bisimilarity (full branching) is sometimes considered
too strict, while trace equivalence (no
branching) is often considered too coarse. The \emph{linear time / branching time spectrum}
~\cite{Glabbeek90} shows a taxonomy of many interesting equivalences lying in between bisimilarity and traces.

Labeled transition system are coalgebras for the functor $\pow(Id)^A$ and the coalgebraic equivalence $\sim_{\pow(Id)^A}$ coincides
with the standard notion of Park-Milner bisimilarity. In~\cite{PowerTuri}, it is shown a coalgebraic
characterization of traces semantics (for LTS's) employing Kleisli categories. 
More recently,~\cite{Monteiro08} have provided a characterization of trace, failure and
ready semantics by mean of ``behaviour objects''. Another coalgebraic approach~\cite{Klin04} relies on ``test-suite'' that, intuitively, 
are fragments of Hennessy-Milner logic.
In this section, we show that (finite) trace equivalence~\cite{Hoare78}, complete trace equivalence~\cite{Glabbeek90},
failures~\cite{BrookesHR84} and ready semantics~\cite{OlderogH86}
can be seen as special cases of $\approx_F^T$.

\bigskip

Before introducing these semantics, we fix some notations. A
labeled transition system is a pair $(X, \delta)$ where $X$ is a set
of states and $\delta\colon X \to \pow(X)^A$ is a function assigning to
each state $x\in X$ and to each label $a\in A$ a finite set of possible
successors states: $x\tr{a}y$ means that $y\in \delta(x)(a)$. Given a
word $w \in A^*$, we write $x\tr{w}y$ for $x\tr{a_1}\dots \tr{a_n}y$
and $w=a_1 \dots a_n$. When $w=\epsilon$, $x\tr{\epsilon}y$ iff
$y=x$.
For a function $\varphi\in \pow(X)^A$, $I(\varphi)$ denotes the set
of all labels ``enabled'' by $\varphi$, i.e., $\{a\in A \mid
\varphi(a)\neq \emptyset\}$, while $\textit{Fail}(\varphi)$ denotes the set
$\{Z \subseteq A \mid Z \cap I(\varphi)=\emptyset\}$.

%Given a LTS $<X,\delta>$ and an $x\in X$, we will write $I(x)$ as a
%shorthand for $I(\delta(x))$.

Let $<X, \delta>$ be a LTS and $x\in X$ be a state. A \emph{trace}
of $x$ is a word $w \in A^*$ such that $x \tr{w}y$ for some $y$. A
trace $w$ of $x$ is \emph{complete} if $x\tr{w}y$ and $y$ stops,
i.e., $I(\delta(y))=\emptyset$. A \emph{failure pair} of $x$ is a
pair $<w, Z>\in A^*\times \pow(A)$ such that $x\tr{w}y$ and $Z\in
\textit{Fail}(\delta(y))$. A \emph{ready pair} of $x$ is a pair $<w, Z>\in A^*\times \pow(A)$ such
that $x\tr{w}y$ and $Z = I(\delta (y))$. We use $\mathcal{T}(x)$,
$\mathcal{CT}(x)$, $\mathcal{F}(x)$ and $\mathcal{R}(x)$ to denote,
respectively, the set of all traces, complete traces, failure pairs
and ready pairs of $x$. For $\mathcal{I}$ ranging over $\mathcal{T},
\mathcal{CT}, \mathcal{F}$ and $\mathcal{R}$, two states $x$ and $y$
are $\mathcal{I}$-equivalent iff $\mathcal{I}(x)=\mathcal{I}(y)$.

\bigskip

For an example, consider the following transition systems
labeled over $A=\{a,b,c\}$. They are all trace equivalent because
their traces are $a,ab,ac$. The trace $a$ is also complete for $p$,
but not for the others. Only $r$ and $s$ are failure equivalent,
since $<a,\{bc\}>$ is a failure pair only of $p$, while $<a,\{b\}>$
and $<a,\{c\}>$ are failure pairs of $p$, $r$ and $s$, but not of
$q$. Finally they are all ready different, since $<a,\emptyset>$ is
a ready pair only of $p$, $<a,\{b,c\}>$ is a ready pair of $q$ and
$s$ but not of $r$, and $<a,\{b\}>$ and $<a,\{c\}>$ are ready pairs
only of $r$ and $s$.
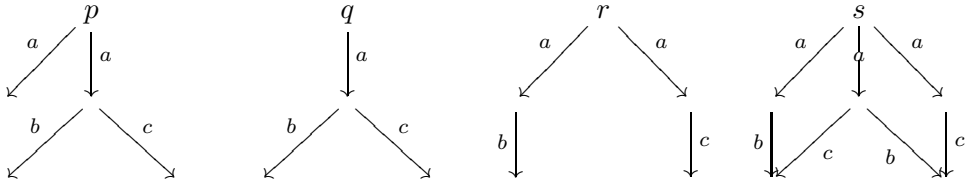
\[
\xymatrix{& p \ar[d]^{a} \ar[dl]_{a} & && q \ar[d]^{a} & && r \ar[dr]^{a}\ar[dl]_{a} & && s \ar[d]|a \ar[dr]^{a}\ar[dl]_{a} \\
 &\ar[dr]^{c}\ar[dl]_{b}& & & \ar[dr]^{c}\ar[dl]_{b}& & \ar[d]_{b}&& \ar[d]^{c} &   \ar[d]_b& \ar[rd]_b\ar[ld]^c& \ar[d]^c\\
 &&  & && &&  & &&& }
\]

We can now show that these equivalences are instances of
$\approx^T_F$. We first show ready equivalence in details and
then, briefly, the others.

Take $T=\pow$ and $F=\pow(\pow(A)) \times id^A$. For each set $X$,
consider the function $\pi^\mathcal{R}_X \colon \pow(X)^A \to FT(X)$
defined for all $\varphi \in \pow(X)^A$ by
\[\pi^\mathcal{R}_X(\varphi)=<\{I(\varphi)\}, \varphi>\text{.}\] This function allows to
transform each LTS $(X,\delta)$ into the $FT$-coalgebra
$(X,\pi^\mathcal{R}_X \circ \delta)$. The latter has the same
transitions of $<X,\delta>$, but each state $x$ is ``decorated''
with the set
$\{I(\varphi)\}$.%\footnote{The family of functions
%$\pi^\mathcal{R}_X:\pow(X)^A \to FT(X)$ forms an injective natural
%transformation $\pi^\mathcal{R}:\pow(-)^A \Rightarrow FT$. The same
%hold for $\pi^\mathcal{F}_X$, $\pi^\mathcal{T}_X$ and
%$\pi^\mathcal{TC}_X$ that we will introduce afterwards.}

Now, by employing the powerset construction, we transform
$<X,\pi^\mathcal{R}_X\circ \delta>$ into the $F$-coalgebra
$(\pow(X), <o,t> )$, where, for all $Y \in \pow(X)$, $a\in A$, the
functions $o \colon \pow(X)\to \pow(\pow(A))$ and $t \colon \pow(X)\to \pow(X)^A$
are
$$o (Y) = \bigcup\limits_{y\in Y}\{I(\delta(y))\}  \qquad t(Y)(a) =
\bigcup\limits_{y\in Y}\delta(y)(a)\text{.}$$

The final $F$-coalgebra is $(\pow(\pow(A))^{A^*},<\epsilon, (-)_a>)$
where $<\epsilon, (-)_a>$ is defined as usual.
\[
\xymatrix{ X \ar[d]_{\delta}\ar[r]^{\{\cdot\}} & \pow(X)\ar@{}[ddrrr]|{\small\begin{array}{l}
\bb{Y}(\epsilon) = o(Y)\\
\bb{Y}(aw) = \bb{t(Y)(a)}(w)\\
\end{array}  }
\ar[ddl]^-{<o,t>}\ar@{-->}[rrr]^{\bb{-}} &&& \pow(\pow(A))^{A^*}
\ar[dd]^{<\epsilon, (-)_a>}\\
(\pow(X))^A \ar[d]_{\pi^\mathcal{R}_X} & & &
\\
\pow(\pow(A)) \times (\pow(X))^A\ar@{-->}[rrrr] &&&&
\pow(\pow(A))\times (\pow(\pow(\pow(A))^{A^*}))^A
} %
\]
Summarizing, the final map
$\bb{-} \colon \pow(X) \to \pow(\pow(A))^{A^*}$ maps each $\{x\}$ into a
function assigning to each word $w$, the set $\{Z\subseteq A \mid x\tr{w}y \text{ and } Z=I(\delta(y))\}$. In other terms,
$Z\in \bb{\{x\}}(w)$ iff $<w,Z>\in \mathcal{R}(x)$.

For the state $s$ depicted above,
$\bb{\{s\}}(\epsilon)=\{\{a\}\}$, $\bb{\{s\}}(a)=\{\{b\}, \{b,c\},
\{c\}\}$, $\bb{\{s\}}(ab)=\bb{\{s\}}(ac)=\{\emptyset\}$ and for all
the other words $w$, $\bb{\{s\}}(w)=\emptyset$.

\bigskip

The other semantics can be characterized in the same way, by
choosing different functors $F$ and different functions
$\pi_X \colon \pow(X)^A \to FT$.

For failure semantics, take the same functor as for the ready semantics, that is $F=\pow(\pow(A)) \times id^A$ 
and a new function $\pi^\mathcal{F}_X \colon \pow(X)^A \to FT(X)$
defined $\forall \varphi \in \pow(X)^A$ by $$\pi^\mathcal{F}_X (\varphi)= <\textit{Fail}(\varphi),
\varphi>\text{.}$$
%$\pi^F:\pow(-)^A \Rightarrow \pow(\pow(A)) \times \pow(-)^A$ is the
%family of functions
%$\pi^F_X:\pow(X)^A \to \pow(\pow(A)) \times \pow(X)^A$ defined
%$\forall \varphi \in \pow(X)^A$ as $<F, \varphi>$ where $F=\{X
%\text{ s.t.} X\cap I(\varphi)=\emptyset\}$.
The $FT$-coalgebra $(X,\pi^\mathcal{F}_X \circ \delta)$ has the same
transitions of the LTS $<X,\delta>$, but each state $x$ is ``decorated''
with the set $\textit{Fail}(\varphi)$.

For both trace and complete trace equivalence, take $F=2 \times
id^A$ (as for NDA). For trace equivalence,
$\pi^\mathcal{T}_X \colon \pow(X)^A \to FT(X)$ maps $\varphi \in \pow(X)^A$
into $<1, \varphi>$. Intuitively, $(X,\pi^\mathcal{T}_X \circ \delta)$ is
an NDA where all the states are accepting.
For complete traces, $\pi^{\mathcal{CT}}_X\colon \pow(X)^A \to FT(X)$ maps
$\varphi$ in $<1,\varphi>$ if $I(\varphi)=\emptyset$ (and in
$<0,\varphi>$ otherwise).

By taking $T=\mathcal D_\omega$ instead of $T=\mathcal P_\omega$, we hope to be able to 
 characterize probabilistic trace, complete trace, ready and
failure as defined in \cite{JouS90}.

\section{Discussion}\label{sec:discussion}
In this paper, we lifted the powerset construction on automata to the
more general framework of $FT$-coalgebras.
Our results lead to a uniform treatment of several kinds of
existing and new variations of automata (that is,
$FT$-coalgebras) by an algebraic structuring of their state space through a monad $T$.
 We showed as examples partial Mealy
machines, structured Moore automata, nondeterministic, partial and
 probabilistic automata. Furthermore, we have presented an interesting coalgebraic characterization of pushdown automata and showed how several behavioural equivalences stemming from 
 concurrency theory can be retrieved from the general framework. It is worth
mentioning that the framework instantiates to many other examples,
among which are \emph{weighted automata}~\cite{Schutzenberger61b}. These
are simply structured Moore automata for $B=1$ and
$\T=\mathbb{S}_{\omega}^-$ (for a semiring
$\mathbb{S}$)~\cite{gumm}.
It is easy to see that $\sim_{FT}$ coincides with weighted
bisimilarity~\cite{german}, while $\approx^T_F$ coincides with
weighted language equivalence~\cite{Schutzenberger61b}.

\medskip

Some of the aforementioned examples can also be coalgebraically
characterized in the framework of~\cite{HJS,HasuoThesis}. There, instead of
considering $FT$-coalgebras on $\Set$ and $F^*$-coalgebras on
$\Set^{\T}$ (the Eilenberg-Moore category), $TG$-coalgebras on
$\Set$ and $\overline{G}$-coalgebras on $\Set_{\T}$ (the
\emph{Kleisli} category) are studied. The main theorem of~\cite{HJS}
states that under certain assumptions, the initial $G$-algebra is
the final $\overline{G}$-coalgebra that characterizes (generalized)
trace equivalence. %The relationship between both framework is not clear yet and further research
%is needed in order to make it precise. 
The exact relationship between these two approaches has been studied in~\cite{JSS} 
(and, indirectly, it could be deduced from~\cite{BK} and~\cite{KK}). 
It is  worth to remark that many of our examples do not fit the framework in~\cite{HJS}: for instance, the exception, the
side effect, the full-probability and the interactive output monads do not fulfill their
requirements (the first three do not have a bottom element and the
latter is not commutative). Moreover, we also note
that the example of partial Mealy machines is not purely trace-like,
as all the examples in~\cite{HJS}.

%under certain assumptions, the $\overline{G}$-coalgebras of~\cite{HJS} can be translated (preserving the semantics) 
%into $F^*$-coalgebras.

The idea of using monads for modeling automata with non-determinism, 
probabilism or side-effects dates back to the ``$\lambda$-machines'' of~\cite{AM75} that, rather than coalgebras, rely on algebras. 
More precisely, the dynamic of a $\lambda$-machine is a morphism 
$\delta \colon FX \to TX$, where $F$ is a functor and $T$ is a monad (for instance the transitions of $T$-structured Moore automata are 
a function $\delta \colon X \times A \to TX$ mapping a state and an input symbol into an element of $TX$).
Analogously to our approach, each $\lambda$-machine induces an ``implicit $\lambda$-machine'' having $TX$ as state space. Many examples of this paper (like Moore automata) %that, by currying and uncurrying, are both algebras and coalgebras) 
can be seen as $\lambda$-machines, but those systems that are essentially coalgebraic (like Mealy machines) 
do not fit the framework in~\cite{AM75}. %On the other hand, $\lambda$-machines are equipped also with an initial state and, therefore,
%they cannot be naively included in our framework.

\medskip

There are several directions for future research. On the one hand, we
will try to exploit \emph{$F$-bisimulations up to $T$}~\cite{Lenisa99,LenisaPW00}
as a sound and complete proof technique for $\approx_F^T$.
%as a we will try to incorporate comonadic %omputations~\cite{UustaluV08}.
%
On the other hand, we would like to lift many of those coalgebraic
tools that have been developed for ``branching equivalences'' (such
as coalgebraic modal logic \cite{ml,Lutz_expressivity} and
(axiomatization for) regular expressions~\cite{BRS09b}) to work with
the ``linear equivalences'' induced by $\approx_F^T$.

We have pursued further the applications to decorated traces and the challenging modeling of the full linear-time spectrum in a separate paper~\cite{mfps}, work which we also plan to extend to probabilistic traces. 

%\bibliographystyle{plain}
%\bibliography{refs}

\begin{thebibliography}{10}

\bibitem{Ada74}
J.~Ad\'{a}mek.
\newblock Free algebras and automata realization in the language of categories.
\newblock {\em Comment. Math. Univ. Carolinae}, 15:589--602, 1974.

\bibitem{AM75}
M. Arbib, and E. Manes.
\newblock Fuzzy machines in a category.
\newblock {\em Bull. Austral. Math. Soc.}, 13:169--210, 1975.

\bibitem{ABB97}
J.-M. Autebert, J. Berstel, and L. Boasson.
\newblock Context-Free Languages and Push-Down Automata.
\newblock In G. Rozenberg and  A. Salomaa (eds.), {\em Handbook of Formal Languages},
 Volume 1, pages 111-174. Springer-Verlag, 1997.

\bibitem{bartels}
F. Bartels.
\newblock {\em On generalized coinduction and probabilistic specification formats}.
\newblock PhD thesis, Vrije Universiteit Amsterdam, 2004.

\bibitem{BHM00}
N. Benton, J. Hughes, and E. Moggi.
\newblock Monads and Effects.
\newblock Course notes for {\em APPSEM Summer School}, 2000. Available on line at
  {\tt http://www.disi.unige.it/person/MoggiE/APPSEM00/BHM.ps}.

\bibitem{BK}
A. Balan, and A. Kurz.
\newblock On Coalgebras over Algebras. 
\newblock {\em Electronic Notes in Theoretical Computer Science}. 264(2): 47-62 (2010)

\bibitem{mfps}
F. Bonchi, M.M. Bonsangue, G. Caltais, J.J.M.M. Rutten, and A. Silva.
\newblock Final semantics for decorated traces,
\newblock In {\em Proceedings of MFPS}, ENTCS, Elsevier, 2012, to appear.

\bibitem{BRS09b}
M.M. Bonsangue, J.J.M.M. Rutten, and A. Silva.
\newblock An algebra for {K}ripke polynomial coalgebras.
\newblock In {\em Proceedings of 24th Annual IEEE Symposium on Logic In Computer
   Science (LICS 2009)}, pages 49--58. IEEE Computer Society, 2009.

\bibitem{BrookesHR84}
S.D. Brookes, C.A.R. Hoare and A.W. Roscoe.
\newblock A Theory of Communicating Sequential Processes.
\newblock {\em Journal of the ACM}, 31(3):560--599, ACM 1984.

\bibitem{german}
P. Buchholz.
\newblock Bisimulation relations for weighted automata.
\newblock {\em Theoretical Computer Science}, 393(1-3):109--123, Elsevier, 2008.

\bibitem{Cho62}
N. Chomsky.
\newblock Context Free Grammars and Pushdown Storage.
\newblock {\em Quarterly Progress Report}, volume 65, MIT Research Laboratory
   in Electronics, Cambridge, MA, 1962.

\bibitem{ml}
C. C\^{\i}rstea, A. Kurz, D. Pattinson, L. Schr{\"o}der, and Y. Venema.
\newblock Modal logics are coalgebraic.
\newblock {\em Computer Journal} 54(1):31--41, Oxford University Press, 2011.

\bibitem{Eve63}
R.J. Evey.
\newblock Application of Pushdown Store Machines.
\newblock In {\em Proceedings of the 1963 Fall Joint Computer Conference (AFIPS 1963)}, ACM, 1963.

\bibitem{Glabbeek90}
R.J. van Glabbeek.
\newblock The Linear Time-Branching Time Spectrum.
\newblock In E. Best (Ed.), {\em Proceedings of CONCUR 93}, volume 458 of
    {\em Lecture Notes in Computer Science}, pages 278--297. Springer, 1990.

\bibitem{Gre67}
S. Greibach.
\newblock A Note on Pushdown Store Automata and Regular Systems.
\newblock {\em Proceedings of the American Mathematical Society}, 18:263--268, American
  Mathematical Society 1967.

\bibitem{gumm}
H.P. Gumm and T. Schr{\"o}der.
\newblock Monoid-labeled transition systems.
\newblock {\em Electronic Notes in Theoretical Computer Science}, 44(1):184--203, Elsevier 2001.

\bibitem{helle}
H.H. Hansen.
\newblock Coalgebraising subsequential transducers.
\newblock {\em Electronic Notes in Theoretical Computer Science}, 203(5):109--129, 2008.

\bibitem{HasuoThesis}
I. Hasuo.
\newblock {\em Tracing Anonymity with Coalgebras}.
\newblock PhD thesis, Radboud University Nijmegen, 2008.

\bibitem{HJS}
I. Hasuo, B. Jacobs, and A. Sokolova.
\newblock Generic trace semantics via coinduction.
\newblock {\em Logical Methods in Computer Science}, 3(4):1--36, 2007.

\bibitem{Hoare78}
C. A. R. Hoare.
\newblock Communicating Sequential Processes.
\newblock {\em Communincation of the ACM.}, 21(8):666--677, ACM, 1978.

\bibitem{HU79}
J. Hopcroft, J. Ullman.
\newblock \emph{Introduction to Automata Theory, Languages, and Computation}.
\newblock Addison-Wesley, 1979.

\bibitem{jacobs05}
B. Jacobs.
\newblock Distributive laws for the coinductive solution of recursive equations.
\newblock {\em Information and Computation}, 204(4): 561-587, 2006.

\bibitem{JSS}
B. Jacobs, A. Silva, and A. Sokolova.
\newblock Trace Semantics via Determinization.
\newblock To appear in {\em Proceedings of CMCS 12}, in {\em Lecture Notes in Computer Science}. Springer, 2012. 

\bibitem{Joh75}
P.T. Johnstone.
\newblock Adjoint lifting theorems for categories of algebras.
\newblock {\em Bulletin London Mathematical Society}, 7:294--297, 1975.

\bibitem{JouS90}
C. Jou and S.A. Smolka.
\newblock Equivalences, Congruences, and Complete Axiomatizations for Probabilistic Processes.
\newblock In J. Baeten and J.W. Klop (eds), {\em  proceedings of CONCUR '90}, volume 458 of
  {\em Lecture Notes in Computer Science}, pages 367--383, Springer, 1990.

\bibitem{Klin04}
B. Klin.
\newblock A coalgebraic approach to process equivalence and a coinduction principle for traces.
\newblock \newblock {\em  Electronic Notes in Theoretical Computer Science}, 106:201--218, 2004.

\bibitem{KK}
C. Kissig, and A. Kurz. 
\newblock Generic Trace Logics.
\newblock In {\em arXiv:1103.3239v1 [cs.LO]}, 2011.

\bibitem{LarsenS91}
K.G. Larsen and A. Skou.
\newblock Bisimulation through probabilistic testing.
\newblock {\em Information and Computation}, 94(1):1--28, 1991.

\bibitem{Lenisa99}
M. Lenisa.
\newblock From Set-theoretic Coinduction to Coalgebraic Coinduction: some results, some problems.
\newblock {\em Electronic Notes in Theoretical Computer Science}, 19:2--22, Elsevier, 1999.

\bibitem{LenisaPW00}
M. Lenisa, J. Power and H. Watanabe.
\newblock Distributivity for endofunctors, pointed and co-pointed endofunctors, monads and comonads.
\newblock {\em Electronic Notes in Theoretical Computer Science}, 33:230--260, Elsevier, 2000.

%\bibitem{ML71}
%S. Mac Lane.
%\newblock {Categories for the Working Mathematician}.
%\newblock {\em Graduate Texts in Mathematics}, 5, Springer 1971.

\bibitem{Man76}
E. Manes.
\newblock {\em Algebraic theories}.
\newblock {\em Graduate Texts in Mathematics}, 26, Springer 1976.

\bibitem{Moggi}
E. Moggi.
\newblock Notions of computation and monads.
\newblock {\em Information and Computation}, 93(1):55--92, 1991.

\bibitem{Monteiro08}
L. Monteiro.
\newblock A Coalgebraic Characterization of Behaviours in the Linear Time - Branching Time Spectrum.
\newblock In {\em proceedings of the 19th International Workshop on Recent Trends in Algebraic
  Development Techniques (WADT 2008)}, volume 5486 of {\em Lecture Notes in Computer Science},
  pages 128--140. Springer, 2009.

\bibitem{OlderogH86}
E.-R. Olderog and C.A.R. Hoare.
\newblock Specification-Oriented Semantics for Communicating Processes.
\newblock {\em Acta Informaticae}, 21(1):9--66, 1986.

\bibitem{PowerTuri}
J. Power and D. Turi.
\newblock A Coalgebraic Foundation for Linear Time Semantics.
\newblock {\em  Electronic Notes in Theoretical Computer Science}, 160:305--29, 1999.

\bibitem{rabin}
M.O. Rabin.
\newblock Probabilistic automata.
\newblock {\em Information and Control}, 6(3):230--245, 1963.

\bibitem{Rutten00}
J.J.M.M. Rutten.
\newblock Universal coalgebra: a theory of systems.
\newblock {\em Theoretical Computer Science}, 249(1):3--80, Elsevier, 2000.

\bibitem{jan_mealy}
J.J.M.M. Rutten.
\newblock Algebraic specification and coalgebraic synthesis of mealy automata.
\newblock {\em Electronic Notes in Theoretical Computer Science}, 160:305--319, 2006.

\bibitem{jan99}
J.J.M.M. Rutten.
\newblock Coalgebra, concurrency, and control.
\newblock In R.~Boel and G.~Stremersch (eds.), {\em proceedings of the 5th Workshop
  on Discrete Event Systems (WODES 2000)}, pages 31--38,  Kluwer, 2000.

\bibitem{Lutz_expressivity}
L. Schr{\"o}der.
\newblock Expressivity of coalgebraic modal logic: The limits and beyond.
\newblock {\em Theoretical Computer Science}, 390(2-3):230--247, Elsevier, 2008.

\bibitem{Schutzenberger61b}
M.P. Sch{\"u}tzenberger.
\newblock On the definition of a family of automata.
\newblock {\em Information and Control}, 4(2-3):245--270, 1961.

\bibitem{Sch63}
M.P. Sch\"{u}tzenberger.
\newblock On Context Free Languages and Pushdown Automata.
\newblock {\em Information and Control}, 6:246-264, 1963.

\bibitem{FSTTCS}
A. Silva, F. Bonchi, M. Bonsangue and J. Rutten.
\newblock Generalizing the powerset construction, coalgebraically.
\newblock In proceedings of {\em IARCS Annual Conference on Foundations
  of Software Technology and Theoretical Computer Science, (FSTTCS 2010)},
  volume 8 of {\em LIPIcs}, pages 272 -- 283, Schloss Dagstuhl - Leibniz-Zentrum
  fuer Informatik, 2010.


\bibitem{WBR11}
J. Winter, M.M. Bonsangue, J.J.M.M. Rutten.
\newblock Context-Free Languages, Coalgebraically.
\newblock In A. Corradini, B. Klin, and C. Cirstea, (eds.), {\em Proceedings of
   4th Int. Conference on Algebra and Coalgebra in Computer science (CALCO 2011)},
   volume 6859 of {\em Lecture Notes in Computer Science}, pages 359-376, Springer, 2011.

\end{thebibliography}

\end{document}